\documentclass[final]{article}
\usepackage[utf8]{inputenc}
\usepackage{natbib}
\usepackage[top=0.8in, bottom=1in, left=1.25in, right=1.25in]{geometry}
\usepackage{scrextend}
\usepackage[utf8]{inputenc} 
\usepackage[T1]{fontenc}    
\usepackage[english]{babel}
\usepackage{lmodern}
\usepackage[bitstream-charter]{mathdesign}
\usepackage[table]{xcolor}
 \usepackage{hyperref}
 \definecolor{burgundy}{rgb}{0.5, 0.0, 0.13}
\definecolor{camel}{rgb}{0.76, 0.6, 0.42}
\definecolor{chamoisee}{rgb}{0.63, 0.47, 0.35}
\definecolor{grey1}{RGB}{128,128,128}
 \hypersetup{colorlinks=true,linkcolor=burgundy,citecolor=chamoisee,urlcolor=burgundy,linktoc=page}
\usepackage{url}            
\usepackage{booktabs}       
\usepackage{nicefrac}       
\usepackage{microtype}      
\usepackage{dsfont}         
\renewcommand{\epsilon}{\varepsilon}
\renewcommand{\phi}{\varphi}

\newcommand\crule[3][black]{\textcolor{#1}{\rule{#2}{#3}}}

\title{\vspace*{0cm}\bf Random Geometric Graph: Some recent developments and perspectives}
\newcommand{\qu}[1]{\textcolor{red}{#1}}

\usepackage{ulem}

\usepackage{amsthm}
\usepackage{forest,adjustbox}
\usetikzlibrary{ chains, scopes} 

\author{}

\date{}

\usepackage{enumitem}   
\usepackage{amsmath}
\newtheorem{theorem}{Theorem}
\newtheorem{conjecture}{Conjecture}
\newtheorem{Lemma}{Lemma}
\newtheorem{proposition}{Proposition}

\newtheorem{definition}{Definition}
\usepackage{mathtools}
\DeclarePairedDelimiter\ceil{\lceil}{\rceil}
\DeclarePairedDelimiter\floor{\lfloor}{\rfloor}
\usepackage{diagbox}
\usepackage{subcaption}
\usepackage{algorithm,refcount}
\usepackage[noend]{algorithmic}
\usepackage{tikz}
\usetikzlibrary{%
  arrows,%
  calc,%
  shapes.geometric,%
  shapes.misc,%
  shapes.symbols,%
  shapes.arrows,%
  automata,%
  through,%
  positioning,%
  scopes,%
  decorations.shapes,%
  decorations.text,%
  decorations.pathmorphing,%
  shadows}
\usetikzlibrary{positioning}
\newcolumntype{P}[1]{>{\centering\arraybackslash}p{#1}}
\usepackage{multirow}
\usepackage{makecell}
\usepackage{mdframed}

\usepackage{savesym}
\savesymbol{checkmark}
\usepackage{dingbat}
\usepackage{wasysym}
\usepackage{appendix}

\begin{document}

\maketitle

{\bf Quentin Duchemin and Yohann De Castro}
\bigskip

\begin{tabular}{m{1.3cm}m{11cm}}
{\bf Keywords} &Random Geometric Graphs~$\bullet$ Concentration inequality for U-statistics~$\bullet$ Random matrices~$\bullet$ Non-parametric estimation~$\bullet$ Spectral clustering~$\bullet$ Coupling~$\bullet$ Information inequalities
\end{tabular}

\begin{abstract}

The Random Geometric Graph (RGG) is a random graph model for network data with an underlying spatial representation. Geometry endows RGGs with a rich dependence structure and often leads to desirable properties of real-world networks such as the small-world phenomenon
and clustering. Originally introduced to model wireless communication networks, RGGs are now very popular with applications ranging from network user profiling to protein-protein interactions in biology. RGGs are also of purely theoretical interest since the underlying geometry gives rise to challenging mathematical questions. Their resolutions involve results from probability, statistics, combinatorics or information theory, placing RGGs at the intersection of a large span of research communities.
 \\
This paper surveys the recent developments in RGGs from the lens of high dimensional settings and non-parametric inference. We also explain how this model differs from classical community based random graph models and we review recent works that try to take the best of both worlds. As a by-product, we expose the scope of the mathematical tools used in the proofs.
\end{abstract}

\section{Introduction}

\subsection{Random graph models}

Graphs are nowadays widely used in applications to model real world complex systems. Since they are high dimensional objects, one needs to assume some structure on the data of interest to be able to efficiently extract information on the studied system. To this purpose, a large number of models of random graphs have been already introduced. The most simple one is the Erdös-Renyi model~$G(n,p)$ in which each edge between pairs of~$n$ nodes is present in the graph with some probability~$p \in (0,1)$. One can also mention the scale-free network model of Barabasi and Albert \citep{barabasi09} or the small-world networks of Watts and Strogatz \citep{watts98}. We refer to \cite{Channarond15} for an introduction to the most famous random graph models. On real world problems, it appears that there often exist some relevant variables accounting for the heterogeneity of the observations. Most of the time, these explanatory variables are unknown and carry a precious information on the system studied. To deal with such cases, latent space models for network data emerged (see \cite{smith19}). Ones of the most studied latent models are the {\it community based random graphs} where each node is assumed to belong to one (or multiple) community while the connection probabilities between two nodes in the graph depend on their respective membership. The well-known Stochastic Block Model has received increasing attention in the recent years and we refer to \cite{abbe18} for a nice introduction to this model and the statistical and algorithmic questions at stake. In the previous mentioned latent space models the intrinsic geometry of the problem is not taken into account. However, it is known that the underlying spatial structure of network is an important property since geometry affects drastically the topology of networks (see \cite{Barthelemy11} and  \cite{smith19}). To deal with embedded complex systems, spatial random graph models have been studied such as the Random Geometric Graph (RGG). This paper surveys the recent developments in the theoretical analysis of RGGs through the prism of modern statistics and applications.

The theoretical analysis of random graph models is interesting by itself since it often involves elegant and important information theoretic, combinatorial or probabilistic tools. In the following, we adopt this mindset trying to provide a faithful picture of the state of the art results on RGGs focusing mainly on high dimensional settings and non-parametric inference while underlining the main technical tools used in the proofs. We want to illustrate how the theory can impact real data applications. To this end, we will essentially be focused on the following questions: 
\begin{itemize}
\item {\bf Detecting Geometry in RGGs.} Nowadays real world problems often involve high-dimensional feature spaces. A first natural work is to identify the regimes where the geometry is lost in the dimension (see Eq.\eqref{eq:geometry-lost} for a formal definition). Several recent papers have made significant progress towards the resolution of this question that can be formalized as follows. Given a graph of~$n$ nodes, a latent geometry of dimension~$d=d(n)$ and edge density~$p=p(n)$, for what triples~$(n,d,p)$ is the model indistinguishable from~$G(n,p)$?
\item {\bf Non-parametric estimation in RGGs.} By considering other rules for connecting latent points, the RGG model can be naturally extended to cover a larger class of networks. In such a framework, we will wonder what can be learned in an adaptive way from graphs with an underlying spatial structure. We will address non-parametric estimation in RGGs and its extension to growth model.
\item {\bf Connections between RGGs and community based latent models.} Until recently, community and geometric based random graph models have been mainly studied separately. Recent works try to investigate graph models that account for both cluster and spatial structures. We present some of them and we sketch interesting research directions for future works.
\end{itemize}

\subsection{Brief historical overview of RGGs}

The RGG model was first introduced by \cite{gilbert61}  to model the communications between radio stations. Gilbert’s original model was defined as follows: pick points in~$\mathds R^2$ according to a Poisson Point Process of intensity one and join two if their distance is less than some parameter~$r>0$.
The Gilbert model has been intensively studied and we refer to \cite{walters11} for a nice survey of its properties including connectivity, giant component, coverage or chromatic number. The most closely related model is the Random Geometric Graph where~$n$ nodes are independently and identically distributed on the space. A lot of results are actually transferable from one model to the other as presented in \cite[Section 1.7]{penrose03}. In this paper we will focus on the~$n$ points i.i.d. model which is formally defined in the next subsection (see Definition~\ref{def:rgg}). The Random Geometric Graph model was extended to other latent spaces such as the hypercube~$[0,1]^d$, the Euclidean sphere or compact Lie group \cite{meliot2019}. A large body of literature has been devoted to studying the properties of low-dimensional Random Geometric Graphs \cite{penrose03}, \cite{dall02}, \cite{bollobas01}. RGGs have found applications in a very large span of fields. One can mention wireless networks \cite{haenggi09}, \cite{Mao12}, gossip algorithms \cite{gang14}, consensus 
\cite{estrada16}, 
spread of a virus \cite{Preciado09}, protein-protein interactions \cite{desmond08}, citation networks \cite{xie16}.
One can also cite an application to motion planning in \cite{solovey18}, a problem which consists in finding a collision-free path for a robot in a workspace cluttered with static obstacles. The ubiquity of this random graph model to faithfully represent real world networks has motivated a great interest for its theoretical study.

\subsection{Outline}

In Section~\ref{sec:models}, we formally define the RGG and several variant models that will be useful for this article. In Sections~\ref{sec:detecting-geometry}, \ref{sec:non-parametric} and \ref{sec:MRGG}, we describe recent results related to high-dimensional statistic, non-parametric estimation and temporal prediction. Note that in these three sections, we will be working with the~$d$-dimensional sphere~$\mathds S^{d-1}$ as latent space.~$\mathds S^{d-1}$ will be endowed with the Euclidean metric~$\|\cdot\|$ which is the norm induced by the inner product~$\langle \cdot,\cdot \rangle : (x,y) \in \left( \mathds S^{d-1}\right)^2 \mapsto \sum_{i=1}^d x_iy_i$. The choice of this latent space is motivated by both recent theoretical developments in this framework \cite{bubeck2016}, \cite{decastro20}, \cite{allen18}, \cite{issartel21} and by applications
\cite{Pereda19}, \cite{perry20}. 
We further discuss in Section~\ref{sec:community-RGG} recent works that investigate the connections between community based random graph models and RGGs. Contrary to the previous sections, our goal is not to provide an exhaustive review of the literature in Section~\ref{sec:community-RGG} but rather to shed light on some pioneering papers.

\begin{center}
\bgroup
\def\arraystretch{1.1}%
\begin{tabular}{c|c|c}
{\bf Section} & {\bf Questions tackled } & {\bf Model}\\\hline \hline
\ref{sec:detecting-geometry} & Geometry detection & RGG on~$\mathds S^{d-1}$\\\hline
\ref{sec:non-parametric} & Non-parametric estimation & TIRGG on~$\mathds S^{d-1}$\\\hline
\ref{sec:MRGG} & \makecell{Non-parametric estimation \\ \& Temporal prediction} & MRGG 
 on~$\mathds S^{d-1}$\\ \hline
 \ref{sec:community-RGG} & \multicolumn{2}{|c}{\makecell{Connections between community\\based models and RGGs}}
\end{tabular}
\captionof{table}{Outline of the paper. Models are defined in Section~\ref{sec:models}.}
\label{table:outline}
\egroup
\end{center}

\section{The Random Geometric Graph model and its variants}
\label{sec:models}

\begin{minipage}{0.42\linewidth}
The questions that we tackle here can require some additional structure on the model. In this section, we define the variants of the RGG that will be useful for our purpose. Figure~\ref{venn-diagram} shows the connections between these different models.

\end{minipage}
\begin{minipage}{0.27\linewidth}
\centering
\includegraphics[scale=0.4]{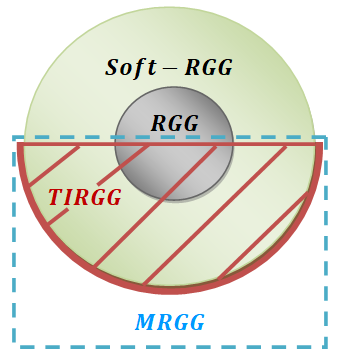}
\label{venn-diagram}
\end{minipage}
\begin{minipage}{0.31\linewidth}
\captionof{figure}{Venn diagram of the different random graph models.}
\end{minipage}

\subsection{(Soft-) Random Geometric Graphs}
\label{general-RGG}

\begin{definition}(Random Geometric Graph: RGG) \\ \label{def:rgg}
Let~$(\mathcal X,\rho)$ be a metric space, and~$m$ be a Borel probability measure on~$\mathcal X$. Given a positive real number~$r>0$, the Random Geometric Graph with~$n\in \mathds N \backslash \{0\}$ points and level~$r>0$ is the random graph~$G$ such that
\begin{itemize}
\item the~$n$ vertices~$X_1,\dots,X_n$ of~$G$ are chosen randomly in~$\mathcal X$ according to the probability measure~$m^{\otimes n}$ on~$\mathcal X^n$.
\item for any~$i,j \in [n]$ with~$i \neq j$, an edge between~$X_i$ and~$X_j$ is present in~$G$ if and only if~$\rho(X_i,X_j)\leq r$.
\end{itemize}
We denote~$\mathrm{RGG}(n,m,(\mathcal X,\rho))$ the distribution of such random graphs.
\end{definition}

Motivated by wireless {\it ad hoc} networks, Soft-RGGs have been more recently introduced (see \cite{Penrose16}). In such models, we are given some function~$H:\mathds R_+ \to [0,1]$ and two nodes at distance~$\rho$ in the graph are connected with probability~$H(\rho)$.

\pagebreak[3]

\begin{definition}(Soft Random Geometric Graph: Soft-RGG) \\ \label{def:soft-rgg}
Let~$(\mathcal X,\rho)$ be a metric space,~$m$ be a Borel probability measure on~$\mathcal X$ and consider some function~$H:\mathds R_+ \to [0,1]$. The Soft (or probabilistic) Random Geometric Graph with~$n\in \mathds N \backslash \{0\}$ points with connection function~$H$ is the random graph~$G$ such that
\begin{itemize}
\item the~$n$ vertices~$X_1,\dots,X_n$ of~$G$ are chosen randomly in~$\mathcal X$ according to the probability measure~$m^{\otimes n}$ on~$\mathcal X^n$.
\item for any~$i,j \in [n]$ with~$i\neq j$, we draw an edge between nodes~$X_i$ and~$X_j$ with probability~$H\left( \rho(X_i,X_j) \right)$.
\end{itemize}
We denote Soft-RGG$(n,m,(\mathcal X,\rho))$ the distribution of such random graphs.
\end{definition}Note that the RGG model with level~$r>0$ is a particular case of the Soft-RGG model where the connection function~$H$ is chosen as~$\rho \mapsto \mathds 1_{\rho \leq r}$. The obvious next special case to consider of Soft-RGG is the so-called percolated RGG introduced in \cite{muller15} which is obtained by retaining each edge of a RGG of level~$r>0$ with probability~$p \in (0,1)$ (and discarding it with probability~$1-p$). This reduces to consider the connection function~$H:\rho \mapsto p \times \mathds 1_{\rho\leq r}$. Particular common choices of connection function are the {\it Reyleigh fading} activation functions which take the form\[H^{Rayleigh}(\rho)= \exp\left[-\zeta\left(  \frac{\rho}{r}\right)^{\eta}\right], \quad \zeta>0, \eta>0.\]  We refer to \cite{Dettmann16} and references therein for a nice overview of Soft-RGGs in particular the most classical connection functions and the question of connectivity in the resulting graphs.

\subsection{Translation Invariant Random Geometric Graphs}

One possible non-parametric generalization of the (Soft)-RGG model is given by the~$W$ random graph model (see for example \cite{diaconis2007graph}) based on the notion of graphon. In this model, given latent points~$x_1,\dots , x_n$ uniformly and independently sampled in $[0, 1]$, the probability to draw an edge between~$i$ and~$j$ is~$\Theta_{i,j}:=W(x_i, x_j )$ where~$W$ is a symmetric function from~$[0, 1]^2$ onto~$[0, 1]$, referred to as a graphon. Hence, the adjacency matrix~$A$ of this graph satisfies \[\forall i,j \in [n], \quad  A_{i,j} \sim \mathrm{Ber}(\Theta_{i,j}),\] where for any~$p \in [0,1]$,~$\mathrm{Ber}(p)$ is the Bernoulli distribution with parameter~$p$. \\
{\bf Remark.} Let us point out that graphon models can also be defined by replacing the latent space $[0,1]$ by the Euclidean sphere $\mathds S^{d-1}:=\{ x \in \mathds R^d\, |\, \|x\|_2=1\}$ in which case latent points are sampled independently and uniformly on $\mathds S^{d-1}$.\\
This model has been widely studied in the literature (see \cite{lovasz12}) and it is now well-known that, by construction, graphons are defined on an equivalent class {\it up to a measure preserving homomorphism}. More precisely, two graphons~$U$ and~$W$ define the same probability distribution if and only if there exist measure preserving maps~$\phi, \psi: [0, 1] \to [0, 1]$ such that~$U (\phi(x), \phi(y)) = W (\psi(x), \psi(y))$ almost everywhere. Hence it can be challenging to have a simple description from an observation given by sampled graph—since one has to deal with all possible composition of a bivariate function by any measure preserving homomorphism. Such difficulty arises in \cite{wolfe13} or in \cite{klopp19} that use respectively Maximum Likelihood and least-square estimators to approximate the graphon~$W$ from the adjacency matrix~$A$. In those works, the error measures are based on the so-called {\it cut-distance} that is defined as an infimum over all
measure-preserving transformations. This statistical issue motivates the introduction of (Soft)-RGGs with latent metric spaces for which the distance is invariant by translation (or
conjugation) of pairs of points. This natural assumption leads to consider that the latent space has some group structure, namely it is a compact Lie group or some compact symmetric space.

\begin{definition}(Translation Invariant Random Geometric Graph: TIRGG) \\ \label{def:tirgg}
Let~$(S,\gamma)$ be a compact Lie group with an invariant Riemannian metric~$\gamma$ normalized
so that the range of~$\gamma$ equals~$[0,\pi]$.
Let~$m$ be the uniform probability measure on~$S$ and let us consider some map~$\mathbf p:[-1,1] \to [0,1]$, called the envelope function. The Translation Invariant Random Geometric Graph with~$n\in \mathds N \backslash \{0\}$ points is the random graph~$G$ such that
\begin{itemize}
\item the~$n$ vertices~$X_1,\dots,X_n$ of~$G$ are chosen randomly in~$S$ according to the probability measure~$m^{\otimes n}$ on~$S^n$.
\item for any~$i,j \in [n]$ with~$i \neq j$, we draw an edge between nodes~$X_i$ and~$X_j$ with probability~$\mathbf p\left( \cos \gamma(X_i,X_j)\right)$.
\end{itemize}
\end{definition}

In Section~\ref{sec:non-parametric}, we present recent results regarding non-parametric estimation in the TIRGG model with~$S:=\mathds S^{d-1}$ the Euclidean sphere of dimension~$d$ from the observation of the adjacency matrix. A related question was addressed in \cite{klopp19} where the authors derived sharp rates of convergence for the~$L^2$ loss for the Stochastic Block Model (which belongs to the class of graphon models). Let us point out that a general approach to control the~$L^2$ loss between the probability matrix and a eigenvalue-tresholded version of the adjacency matrix is the USVT method introduced by \cite{chatterjee15}, which was further investigated by \cite{xu18}. In  Section~\ref{sec:non-parametric}, another line of work is presented to estimate the envelope function~$\mathbf p$ where the difference between the adjacency matrix
and the {\it matrix of probabilities}~$\Theta$ is controlled in operator norm. The cornerstone of the proof is the convergence of the spectrum of the matrix of probabilities towards the spectrum of some integral operator associated with the envelope function~$\mathbf p$. Based on the analysis of \cite{Gine}, the proof of this convergence includes in particular matrix Bernstein inequality from \cite{tropp15} and concentration inequality for order 2 U-statistics with bounded kernels that was first studied by \cite{arcones1993} and remains an active field of research (see \cite{GineZinn}, \cite{houdre2002} or \cite{Lugosi15}).

\subsection{Markov Random Geometric Graphs}

In the following, we will refer to {\it growth models} to denote random graph models in which a node is added at each new time step in the network and is connected to other vertices in the graph according to some probabilistic rule that needs to be specified. In the last decade, growth models for random graphs with a spatial structure have gained an increased interest. One can mention \cite{Jordan15}, \cite{Papadopoulos12} and \cite{zuev15} where geometric variants of the preferential attachment model are introduced with one new node entering the graph at each time step. More recently, \cite{xie15} and \cite{xie16b} studied a growing variant of the RGG model. Note that in the latter works, the birth time of each node is used in the connection function while nodes are still sampled independently in~$\mathds R^2$. Still motivated by non-parametric estimation, the TIRGG model can be extended to a growth model by considering a Markovian sampling scheme of the latent positions. Considering a Markovian latent dynamic can be relevant to model customer behavior for item recommendation or to study bird migrations where animals have regular seasonal
movement between breeding and wintering grounds \citep[cf.][]{ducheminreliable}.

\begin{definition}(Markov Random Geometric Graph: MRGG) \\ \label{def:mrgg}
Let~$(S,\gamma)$ be a compact Lie group with an invariant Riemannian metric~$\gamma$ normalized
so that the range of~$\gamma$ equals~$[0,\pi]$.
Let us consider some map~$\mathbf p:[-1,1] \to [0,1]$, called the envelope function. The Markov Random Geometric Graph with~$n\in \mathds N \backslash \{0\}$ points is the random graph~$G$ such that
\begin{itemize}
\item the sequence of~$n$ vertices~$(X_1,\dots,X_n)$ of~$G$ is a Markov chain on~$S$.
\item for any~$i,j \in [n]$ with~$i \neq j$, we draw an edge between nodes~$X_i$ and~$X_j$ with probability~$\mathbf p\left( \cos \gamma(X_i,X_j)\right)$.
\end{itemize}
\end{definition}

In Section~\ref{sec:MRGG}, we shed light on a recent work from \cite{decastroduchemin20} that achieves non-parametric estimation in MRGGs on the Euclidean sphere of dimension~$d$. The theoretical study of such graphs becomes more challenging because of the dependence induced by the latent Markovian dynamic. Proving the consistency of the non-parametric estimator of the envelope function~$\mathbf p$ proposed in Section~\ref{sec:MRGG} requires in particular a new concentration inequality for U-statistics of order 2 of uniformly ergodic Markov chains. By solving link prediction problems, \cite{decastroduchemin20} also reveal that MRGGs are convenient tools to extract temporal information on growing graphs with an underlying spatial structure.

\subsection{Other model variants}

\paragraph{Choice of the metric space.}
The Euclidean Sphere or the unit square in~$\mathds R^d$ are the most studied latent spaces in the literature for RGGs. By the way, \cite{allen18} offers an interesting comparison of the different topological properties of RGGs working on one or the other of these two spaces. Nevertheless, one can find variants such as in \cite{araya20} where Euclidean Balls are considered. 
More recently, some researchers left the Euclidean case to consider negatively curved--i.e. hyperbolic--latent spaces. Random graphs with an hyperbolic latent space seem promising to faithfully model real world networks. Actually, \cite{krioukov10} showed that the RGG built on the hyperbolic geometry is a scale-free network, that is the proportion of node of degree~$k$ is of order~$k^{-\gamma}$ where~$\gamma~$ is between 2 and 3. The scale-free property is found in the most part of real networks as highlighted by \cite{xie15}.

\paragraph{Different degree distributions.}
It is now well-known that the average degree of nodes in random graph models is a key property for their statistical analysis. Let us highlight some important regimes in the random graph community that will be useful in this paper. The dense regime corresponds to the case where the expected normalized degree of the nodes ({\it i.e.,} degree divided by $n$) is independent of the number of nodes in the graph. The other two important regimes are the relatively sparse and the sparse regimes where the average degree of nodes scales respectively as~$\log(n)/n$ and~$1/n$ with the number of nodes~$n$. A direct and important consequence of these definitions is that in the (relatively) sparse regime, the envelope function $\mathbf p$ from Definitions~\ref{def:tirgg} and~\ref{def:mrgg} depends on $n$ while it remains independent of $n$ in the dense regime. Similarly, in the (relatively) sparse regime, the radius threshold $r$ from Definition~\ref{def:rgg} (resp. the connection function $H$ from Definition~\ref{def:soft-rgg}) depend on $n$ contrary to the dense regime.

\section{Detecting geometry in RGGs}
\label{sec:detecting-geometry}

To quote \cite{bollobas01}, {\it "One of the main aims of the theory of random graphs is to determine when a given property is likely to appear."} In this direction, several works tried to identify structure in networks through testing procedure, see for example \cite{bresler18}, \cite{ghoshdastidar20} or \cite{gao17}. Regarding RGGs, most of the results have been established in the low dimensional regime~$d \leq 3$ \cite{ostilli15}, \cite{Penrose16}, \cite{penrose03}, \cite{Barthelemy11}. \cite{goel05} proved in particular that all monotone graph properties (i.e. property preserved when adding edges to the graph) have a sharp threshold for RGGs that can be distinguished from the one of Erdös-Rényi random graphs in low-dimensions. However, applications of RGGs to cluster analysis and the interest in the statistics of high-dimensional data sets have motivated the community to investigate the properties of RGGs in the case where~$d \to \infty.$ If the ambitious problem of  recognizing if a graph can be realized as a geometric graph is known to be NP-hard \cite{BREU19983}, one can take a step back and wonder if a given RGG still carries some spatial information as~$d \to \infty$ or if geometry is lost in high-dimensions (see Eq.\eqref{eq:geometry-lost} for a formal definition), a problem known as geometry detection. In the following, we present some recent results related to geometry detection in RGGs with latent space the Euclidean sphere $\mathds S^{d-1}$ and we highlight several interesting directions for future research.

\paragraph{Notations} 
Given two sequences~$(a_n)_{n \in \mathds N}$ and~$(b_n)_{n\in \mathds N}$ of positive numbers, we write $a_n=\mathcal O_n(b_n)$ or $b_n = \Omega_n(a_n)$ if the sequence~$(a_n/b_n)_{n\geq 0}$ is bounded and we write~$a_n=o_n(b_n)$ or $b_n=\omega_n(a_n)$ if $a_n/b_n \underset{n\to+\infty}{\rightarrow} 0$. We further denote $a_n = \Theta(b_n)$ if $a_n=\mathcal O_n(b_n)$ and $b_n=\mathcal O_n(a_n)$.
In the following, we will denote~$G(n,p,d)$ the distribution of random graphs of size~$n$ where nodes~$X_1, \dots, X_n$ are sampled uniformly on~$\mathds S^{d-1}$ and where distinct vertices~$i\in[n]$ and~$j\in[n]$ are connected by an edge if and only if~$\langle X_i, X_j\rangle \geq t_{p,d}$. The threshold value~$t_{p,d}\in [-1,1]$ is such that~$\mathds P\left(\langle X_1, X_2\rangle \geq t_{p,d}\right) =p$. Note that~$G(n,p,d)$ is the distribution of RGGs on~$(\mathds S^{d-1},\|\cdot\|)$ sampling nodes uniformly with connection function~$H:t \mapsto \mathds 1_{t\leq \sqrt{2-2t_{p,d}}}.$ In the following, we will also use the notation~$G(n,d,p)$ to denote a graph sampled from this distribution. We also introduce some definitions of standard information theoretic tools with Definition~\ref{def:TV}.

\begin{definition} \label{def:TV}
Let us consider two probability measures~$P$ and~$Q$ defined on some measurable space~$(E,\mathcal E)$. The total variation distance between~$P$ and~$Q$ is given by\[\mathrm{TV}(P,Q) := \sup_{A\in \mathcal E} |P(A)-Q(A)|.\]
Assuming further that~$P\ll Q$ and denoting~$ dP/dQ$ the density of~$P$ with respect to~$Q$,
\begin{itemize}
\item the~$\chi^2$-divergence between~$P$ and~$Q$ is defined by \[\chi^2(P,Q): = \int_E \left( \frac{dP}{dQ}-1\right)^2dQ .\]
\item the Kullback-Leibler divergence between~$P$ and~$Q$ is defined by \[\mathrm{KL}(P,Q): = \int_E  \log\left( \frac{dP}{dQ}\right)dP .\]
\end{itemize}
\end{definition}
Considering that both $p$ and $d$ depend on $n$, we will say in this paper that {\it geometry is lost} if the distributions $G(n,p)$ and $G(n,p,d)$ are indistinguishable, namely if
\begin{equation} \label{eq:geometry-lost}\mathrm{TV}(G(n, p), G(n, p, d))\to0\quad \text{as}~n\to \infty.\end{equation}

\subsection{Detecting geometry in the dense regime}

\cite{devroye2011} is the first to consider the case where~$d \to \infty$ in RGGs. In this paper, the authors proved that the number of cliques in the dense regime in~$G(n,p,d)$ is close to the one of~$G(n,p)$ provided~$d\gg \log n$ in the asymptotic~$d \to \infty$. This work allowed them to show the convergence of the total variation (see Definition~\ref{def:TV}) between RGGs and Erdös-Renyi graphs as~$d \to \infty$ for fixed~$p$ and~$n$. \cite{bubeck2016} closed the question of geometry detection in RGGs in the dense regime showing that a phase transition occurs when~$d$ scales as~$n^3$ as stated by Theorem~\ref{bubeck16-thm1}.

\begin{theorem} \label{bubeck16-thm1}  \cite[Theorem 1]{bubeck2016}
\begin{enumerate}[label=(\roman*)]
\item Let~$p\in (0,1)$ be fixed, and assume that~$d/n^3\to 0$. Then \[\mathrm{TV} (G(n, p), G(n, p, d))\to1\quad \text{as}~n\to \infty.\]
\item Furthermore, if~$d/n^3\to \infty$, then  \[\underset{p \in (0,1)}{\sup} \mathrm{TV}(G(n, p), G(n, p, d))\to 0\quad \text{as}~n\to \infty.\]
\end{enumerate}
\end{theorem}
The proof of Theorem.\ref{bubeck16-thm1}$.(i)$ relies on a count of {\it signed} triangles in RGGs. Denoting~$A$ the adjacency matrix of the RGG, the number of triangles in~$A$ is~$\mathrm{Tr}(A^3)$, while the total number of signed triangles is defined as 
\[\tau(G(n,p,d)) :=\mathrm{Tr}((A-p(J-I))^3)= \sum_{\{i,j,k\} \in \binom{[n]}{3}} \left(A_{i,j}-p \right) \left(A_{i,k}-p \right) \left(A_{j,k} -p\right),\] where~$I$ is the identity matrix and~$J \in \mathds R^{n\times n }$ is the matrix with every entry equals to~$1$. The analogous quantity in Erdös Renyi graphs~$\tau(G(n,p))$ is defined similarly. \cite{bubeck2016} showed that the variance of~$\tau(G(n,p,d))$ is of order~$n^3$ while the one of the number of triangles is of order~$n^4$. This smaller variance for signed triangles is due to the cancellations introduced by the centering of the adjacency matrix. Lemma~\ref{bubeck16-lemma} provides the precise bounds obtained on the expectation and the variance of the statistic of signed triangles. Theorem.\ref{bubeck16-thm1}$.(i)$ follows from the lower-bounds (resp. the upper-bounds) on the expectations (resp. the variances) of $\tau(G(n,p))$ and $\tau(G(n,p,d))$ presented in Lemma~\ref{bubeck16-lemma}.
\begin{Lemma} \label{bubeck16-lemma} \cite[Section 3.4]{bubeck2016} For any~$p \in (0,1)$ and any~$n,d \in \mathds N\backslash \{0\}$ it holds
\begin{align*}&\mathds E\left[ \tau(G(n,p)) \right]=0,  \quad \mathds E\left[ \tau(G(n,p,d)) \right] \geq \binom{n}{3}\frac{C_p}{\sqrt d} \\
\text{and} \quad \max & \left\{ \mathds Var\left[ \tau(G(n,p)) \right],\mathds Var\left[ \tau(G(n,p,d)) \right] \right\} \leq n^3 + \frac{3n^4}{d},
\end{align*}
where $C_p>0$ is a constant depending only on $p$.
\end{Lemma}

Let us now give an overview of the proof of the indistinguishable part of Theorem~\ref{bubeck16-thm1}. \cite{bubeck2016} proved that in the dense regime, the phase transition for geometry detection occurs at the regime at which Wishart matrices becomes indistinguishable from GOEs (Gaussian Orthogonal Ensemble). In the following, we draw explicitly this link in the case~$p=1/2.$\\ An~$n\times n$ Wishart matrix with~$d$ degrees of freedom is a matrix of inner products of~$n$~$d$-dimensional Gaussian vectors denoted by~$W(n, d)$ while an~$n\times n$ GOE random matrix is a symmetric matrix with i.i.d. Gaussian entries on and above the diagonal denoted by~$M(n)$. Let~$\mathds X$ be an~$n\times d$ matrix where the entries are i.i.d. standard normal random variables, and let~$W=W(n, d)=\mathds X \mathds X^{\top}$ be the corresponding~$n\times n$ Wishart matrix. Then recalling that for~$X_1\sim \mathcal N(0,I_d)$ a standard gaussian vector of dimension~$d$,~$X_1/ \|X_1\|_2$ is uniformly distributed on the sphere~$\mathds S^{d-1}$, we get that the~$n\times n$ matrix~$A$ defined by \[\forall i,j \in [n], \quad A_{i,j}=\left\{
    \begin{array}{ll}
        1& \mbox{if } W_{i,j}\geq0 \; \text{and} \; i\neq j \\
        0 & \mbox{otherwise.}
    \end{array}
\right.\]
has the same distribution as the adjacency matrix of a graph sampled from~$G(n,1/2,d)$. We denote~$H$ the map that takes~$W$ to~$A$. Analogously, one can prove that~$G(n,1/2)$ can be seen as a function of an~$n\times n$ GOE matrix. Let~$M(n)$ be a symmetric~$n\times n$ random matrix where the diagonal entries are i.i.d. normal random variables with mean zero and variance~$2$, and the entries above the diagonal are i.i.d.standard normal random variables, with the entries on and above the diagonal all independent. Then~$B=H(M(n))$ is distributed as the adjacency matrix of~$G(n,1/2)$. We then get
\begin{equation} \label{TV-Wishart-GOE}\mathrm{TV} \left(G(n,1/2,d),G(n,1/2)\right) = \mathrm{TV} \left(H(W(n,d)),H(M(n))\right)\leq \mathrm{TV} \left(W(n,d)),M(n)\right).\end{equation}
If a simple application of the multivariate Central Limit Theorem proves that the right hand side of~\eqref{TV-Wishart-GOE} goes to zero as~$d \to \infty~$ for fixed~$n$, more work is necessary to address the case where~$d=d(n)=\omega_n(n^3)$ and~$n\to \infty$. The distributions of~$W(n,d)$ and~$M(n)$ are known and allow explicit computations leading to Theorem~\ref{bubeck16-thm2}. This proof can be adapted for any~$p \in (0,1)$ leading to Theorem~\ref{bubeck16-thm1}.$(ii)$ from~\eqref{TV-Wishart-GOE}.
\begin{theorem} \label{bubeck16-thm2} \cite[Theorem 7]{bubeck2016}

\noindent Define the random matrix ensembles~$W(n, d)$ and~$M(n)$ as above. If~$d/n^3 \to \infty$, then 
\[ \mathrm{TV} (W(n, d),M(n)) \to 0.\]
\end{theorem}

\paragraph{Extensions}

Considering~$\mathds R^d$ as latent space endowed with the Euclidean metric, \cite{bubeck15} extended Theorem~\ref{bubeck16-thm2} and proved an information theoretic phase transition. To give an overview of their result, let us consider the~$n\times n$ Wigner matrix~$\mathcal M_n$ with zeros on the diagonal and i.i.d. standard Gaussians above the diagonal. For some~$n\times d$ matrix~$\mathds X$ with i.i.d. entries from a distribution~$\mu$ on~$\mathds R^d$ that has mean zero and variance 1, we also consider the following rescaled Wishart matrix associated with~$\mathds X$  \[\mathcal W_{n,d} := \frac{1}{\sqrt d} \left(\mathds X \mathds X^{\top} - \mathrm{diag}(\mathds X \mathds X^{\top})\right),\]where the diagonal was removed.
Using an high-dimensional entropic Central Limit Theorem, \cite{bubeck15} proved Theorem~\ref{bubeck15-thm} which implies that geometry is lost in~$RGG(n,\mu,(\mathds R^d,\|\cdot\|))$ as soon as~$d\gg n^3\log^2(d)$ provided that the measure~$\mu$ is sufficiently smooth (namely log-concave) and the rate is tight up to logarithmic factors. We refer to \cite{racz16} for a friendly presentation of this result. Note that the comparison between Wishart and GOE matrices also naturally arise when dealing with covariance matrices. For example, Theorem~\ref{bubeck16-thm2} was used in \cite{brennan19} to study the informational-computational tradeoff of sparse Pincipal Component Analysis.
\begin{theorem} \label{bubeck15-thm} \cite[Theorem 1]{bubeck15}

\noindent If the distribution~$\mu$ is log-concave and~$\frac{d}{n^3\log^2(d)}\to \infty$, then
$\mathrm{TV} (\mathcal W_{n,d},\mathcal M_n)\to 0.$\\On the other hand, if~$\mu$ has a finite fourth moment and~$\frac{d}{n^3}\to 0$, then~$\mathrm{TV} (\mathcal W_{n,d},\mathcal M_n)\to 1.$
\end{theorem}

\subsection{Failure to extend the proof techniques to the sparse regime}

\cite{bubeck2016} provided a result in the sparse regime where~$p =\frac{c}{n}$ showing that one can distinguish between~$G(n,\frac{c}{n})$ and~$G(n,\frac{c}{n}, d)$ as long as~$d \ll \log^3 n$. The authors conjectured that this rate is tight for the sparse regime (see Conjecture~\ref{bubeck16-conjecture}).

\begin{conjecture} \label{bubeck16-conjecture} \cite[Conjecture 1]{bubeck2016}

\noindent Let~$c >0$ be fixed, and assume that~$d/\log^3(n) \to \infty$. Then \[\mathrm{TV}\left(G\left(n,\frac{c}{n}\right), G\left(n,\frac{c}{n}, d\right)\right)\to 0\quad \text{as}~n\to \infty.\]
\end{conjecture}
The testing procedure from \cite{bubeck2016} to prove the distinguishability result in the sparse regime was based on a simple counting of triangles. Indeed, when~$p$ scales as~$\frac{1}{n}$, the signed triangle statistic~$\tau$ does not give significantly more power than the triangle statistic which simply counts the number of triangles in the graph. Recently, \cite{avrachenkov20} provided interesting results that give credit to Conjecture~\ref{bubeck16-conjecture}. First, they proved that in the sparse regime, the clique number of~$G(n, p, d)$ is almost surely at most 3 under the condition~$d \gg \log^{1+\epsilon} n$ for any~$\epsilon>0$. This means that in the sparse regime,~$G(n, p, d)$  does  not  contain any complete subgraph larger than a triangle like Erdös-Renyi graphs. Hence it is hopeless to prove that Conjecture~\ref{bubeck16-conjecture} is false considering the number of~$k$-cliques for~$k\geq 4$. Nevertheless, one could believe that improving the work of \cite{bubeck2016} by deriving sharper bounds on the number of 3-cliques (i.e. the number of triangles), it could possible to statistically distinguish between~$G(n,p,d)$ and~$G(n,p)$ in the sparse regime even for some~$d\gg \log^3 n$. In a regime that can be made arbitrarily close to the sparse one, \cite{avrachenkov20} proved that this is impossible as stated by Theorem~\ref{avrachenkov-thm4}.

\begin{theorem}\label{avrachenkov-thm4}  \cite[Theorem 5]{avrachenkov20}

\noindent Let us suppose that~$d\gg \log^3 n$ and~$p=\theta(n)/n$  with ~$n^m\leq \theta(n)\ll n$ for some~$m>0$. Then the expected number of triangles--denoted $\mathds E[T(n,p,d)]$--in RGGs sampled from~$G(n,p,d)$ is of order~$\binom{n}{3}p^3$, meaning that there exist two universal constants $c,C>0$ such that for $n$ large enough it holds \[c\binom{n}{3}p^3 \leq \mathds E[T(n,p,d)]\leq C\binom{n}{3}p^3  .\]
\end{theorem}
In a nutshell, the work from \cite{avrachenkov20} suggests that a negative result regarding Conjecture~\ref{bubeck16-conjecture} cannot be obtained using statistics based on clique numbers. This discussion naturally gives rise to the following more general question.
\begin{align} &\text{{\it Given a random graph model with~$n$ nodes, latent geometry in dimension~$d=d(n)$ and edge}} \notag\\ &\text{{\it density~$p=p(n)$, for what triples~$(n,d,p)$ is the model~$G(n,p,d)$ indistinguishable from~$G(n,p)$?}} \label{question} \tag{$\mathcal Q$} \end{align} 

\subsection{Towards the resolution of geometry detection}
\label{sec:resolution-geo}

\subsubsection{A first improvement when $d>n$}
A recent work from \cite{brennan20} tackled the general problem~\eqref{question} and proved Theorem~\ref{brennan20-thm}.

\begin{theorem} \label{brennan20-thm} \cite[Theorem 2.4]{brennan20}

\noindent Suppose~$p=p(n)\in(0,1/2]$ satisfies that~$n^{-2}\log n=\mathcal O_n(p)$ and \[ d\gg \min\left\{pn^3\log p^{-1}\;,\;p^2n^{7/2}(\log n)^3\sqrt{\log p^{-1}} \right\},\]where~$d$
 also satisfies that~$d\gg n\log^4n$. Then \[\mathrm{TV}(G(n,p),G(n,p,d))\to 0\quad \text{as}~n\to \infty.\]
\end{theorem}

{\bf Remarks} In the dense regime, Theorem~\ref{brennan20-thm} recovers the optimal guarantee from Theorem~\ref{bubeck16-thm1}. In the sparse regime, Theorem~\ref{brennan20-thm} states that if~$d \gg n^{3/2}\left(\log n\right)^{7/2}$, then geometry is lost in~$G(n,\frac{c}{n},d)$ (where~$c>0$). This result 
improves the work from \cite{bubeck2016}. Nevertheless, regarding Conjecture~\ref{bubeck16-conjecture}, it remains a large gap between the rates~$\log^3 n$ and~$n^{3/2}\left(\log n\right)^{7/2}$ where nothing is known up to date. Let us sketch the main elements of the proof of Theorem~\ref{brennan20-thm}. In the following we denote~$G=G(n,p,d)$ with set of edges~$E(G)$ and for any~$i,j\in [n],\; i \neq j$, we denote~$G_{\sim \{i,j\}}$ the set of edges other than~$\{i,j\}$ in~$G$. One first important step of their approach is the following tenzorization Lemma for the Kullback-Leibler divergence.

\begin{Lemma} \label{brennan20-lemmaKL} \cite[Lemma 3.4]{kontorovitch17}

\noindent Let us consider~$(X,\mathcal B)$ a measurable space with~$X$ a Polish space and~$\mathcal B$ its Borel~$\sigma$-field. Consider some probability measure~$\mu$ on the product space~$X^k$ with~$\mu=\mu_1\otimes \mu_2\otimes \dots \otimes \mu_k$. Then for any other probability measure~$\nu$ on $X^k$ it holds 
\[\mathrm{KL}(\nu||\mu)\leq \sum_{i=1}^k \mathds E_{x \sim \nu}\left[\mathrm{KL}\left(\nu_i(\cdot|x_{\sim i})||\mu_i\right)\right],\]
where $\nu_i$ is the probability distribution corresponding to the $i$-th marginal of $\nu$ and where $x_{\sim i}:=(x_1,\dots,x_{i-1},x_{i+1},\dots,x_k)$.
\end{Lemma}

\begin{align*}
2 \mathrm{TV}(G(n,p,d),G(n,p))^2 &\leq \mathrm{KL}(G(n,p,d)||G(n,p))\qquad \text{from Pinsker's inequality}\\
& \leq \sum_{1\leq i <j \leq n} \mathds E\left[ \mathrm{KL}\left( \mathcal L(\mathds 1_{\{i,j\}\in E(G)}|\sigma(G_{\sim \{i,j\}})) || \mathrm{Bern}(p) \right) \right] \qquad \text{from Lemma~\ref{brennan20-lemmaKL}}\\
&\leq \binom{n}{2}\times  \mathds E\left[ \chi^2\left( \mathcal L(\mathds 1_{e_0\in E(G)}|\sigma(G_{\sim e_0})) ,\mathrm{Bern}(p) \right) \right] \\
&= \binom{n}{2} \times \mathds E\left[ \frac{(Q-p)^2}{p(1-p)}\right],
\end{align*}
where~$Q:=\mathds P\left( e_0 \in E(G) | G_{\sim e_0} \right)$ is a~$\sigma(G_{\sim e_0})$-measurable random variable corresponding to the probability that a specific edge is included in the graph given the rest of the graph. The proof then consists in showing that with high probability,~$Q$ concentrates near~$p$. To do so, they use a coupling argument that gives an alternative way to generate~$X_1$ that provides a direct description of~$\mathds 1_{e_0 \in E(G)}$ in terms of the random variables introduced in the coupling. If this step may seem computationally involved, it is not conceptually difficult since it turns out to be a simple re-parametrization of the problem. An integration of this concentration result for~$Q$ implies that the convergence of Theorem~\ref{brennan20-thm} holds when~$d \gg p n^3 \log p^{-1}$. To get the convergence result in the regime where~$d \gg p^2 n^{7/2}(\log n)^3 \sqrt{\log p^{-1}}$~$-$ which gives the improvement over \cite{bubeck2016} in the sparse case~$-$ one additional step of coupling is required. More precisely, they decompose~$\mathds E[(Q-p)^2]$ as~$\mathds E[(Q-p) \times (Q-p)]$.
The previous coupling argument gives a concentration inequality allowing to bound with high probability the first term $|Q-p|$. It remains then to upper bound~$\mathds E[\left|Q-p\right|]$ which relies on a simple observation given by the following proposition.
\begin{proposition} \label{brennan20-prop53} \cite[Proposition 5.3]{brennan20}
\noindent Let~$\nu_{\sim e_0}$ denote the marginal distribution of~$G$ restricted to all edges that are not~$e_0$, and let~$\nu_{\sim e_0}^+$ denote the distribution of~$G$ conditioned on the event~$e_0\in E(G)$. It holds
\begin{equation}\label{eq-brennan20-prop53} \mathds E[\left|Q-p\right|] = 2p \times \mathrm{TV}\left(\nu_{\sim e_0}^+,\nu_{\sim e_0}\right).\end{equation}
\end{proposition}
The proof is then concluded by using another coupling argument between~$\nu_{\sim e_0}^+$ and~$\nu_{\sim e_0}$ to upper-bound the total variation distance involved in Eq.\eqref{eq-brennan20-prop53} and we give a sketch of proof in the following. Given latent positions $X_1, \dots,X_n$ uniformly and independently sampled on $\mathds S^{d-1}$, we can consider without loss of generality that $X_1=(1,0,\dots0)$. Denoting $X_2=(X_{2,j})_{j \in [d]}$ and $\phi_d$ the density of $X_{2,1}$\footnote{i.e. $\phi_d$ is the density of a one-dimensional marginal of a uniform
random point on $\mathds S^{d-1}$.}, one can define $\gamma = \sqrt{\frac{1-\tau^2}{1-X_{2,1}^2}}$ and $X^+_2:=(\tau, \gamma X_{2,2}, \dots, \gamma X_{2,d})$ where $\tau $ is a random variable in $[-1,1]$ with density $\phi_{d,p}^+(x) = p^{-1} \mathds 1 _{x \geq t_{p,d}} \phi_d(x)$. Denoting further $G_{\sim e_0}$ (resp. $G^+_{\sim e_0}$) the RGG with threshold $t_{p,d}$ induced by the latent points $(X_i)_{i\in  [n]}$ (resp. $(X_1,X^+_2,X_3,\dots,X_n)$)  without the edge $e_0=\{1,2\}$, $G_{\sim e_0}$ (resp.$G^+_{\sim e_0}$) is distributed as $\nu_{\sim e_0}$ (resp. $\nu_{\sim e_0}^+$). Hence it holds,
\begin{align*}\mathds E [ |Q-p|]& \leq 2p \times \mathrm{TV}\left(\nu_{\sim e_0}^+,\nu_{\sim e_0}\right) \leq 2p \times \mathds P (G_{\sim e_0} \neq G^+_{\sim e_0})\leq 2p \sum_{i=3}^n \mathds P\big( \mathds 1_{\langle X_2,X_i \rangle \geq t_{p,d}} \neq \mathds 1_{\langle X_2^+,X_i\rangle \geq t_{p,d}}\big).
\end{align*}
The proof is concluded using standard concentration arguments.

\subsubsection{Reaching the polylogarithmic regime}
Very recently, \cite{Liu2021TestingTF} came with novels ideas and improved upon the previous bounds for geometry detection by polynomial factors in the sparse regime. This significant breakthrough presented in Theorem~\ref{thm:polylog} almost solves Conjecture~\ref{bubeck16-conjecture}.
\begin{theorem}\label{thm:polylog} \cite[Theorem 1.2]{Liu2021TestingTF}
For any fixed constant $c\geq 1$, if $d \gg \log^{36} n$, then
\[\mathrm{TV}\left(G\left(n,\frac{c}{n}\right),G\left(n,\frac{c}{n},d\right)\right) \to 0\quad \text{as}~n\to \infty.\]
\end{theorem}
The authors do not limit their analysis to the sparse regime but also provide results holding for any regime interpolating between the sparse and the dense cases as shown with Theorem~\ref{thm:polylog2}.
\begin{theorem}\label{thm:polylog2} \cite[Theorem 1.1 and Lemma A.1]{Liu2021TestingTF}
\begin{itemize}
\item For any fixed constant $c>0$, if $\frac{c}{n} <p<\frac12$ and $d \gg p^{2} n^3$, then
\[\mathrm{TV}\left(G\left(n,p\right),G\left(n,p,d\right)\right) \to 0\quad \text{as}~n\to \infty.\]
\item If $\frac{1}{n^2} \ll p \leq 1- \delta$ for any fixed constant $\delta> 0$, then as long as $d \ll (nH(p))^3$, 
\[\mathrm{TV}\left(G\left(n,\frac{c}{n}\right),G\left(n,\frac{c}{n},d\right)\right) \to 1\quad \text{as}~n\to \infty,\]
where $H(p) = p\log \frac1p+(1-p) \log \frac1{1-p}$ is the binary entropy function. This result can be achieved using the signed triangle statistic following an approach strictly analogous to \cite{bubeck2016}.
\end{itemize}
\end{theorem}
\cite{Liu2021TestingTF} extend the work from \cite{bubeck2016} and prove that the signed
statistic distinguishes between $G(n, p)$ and $G(n, p,d)$ not only in the sparse and dense cases but also for most $p$, as long as $d \ll (nH(p))^3$. We provide in the Appendix \ref{apdx} a synthetic description of the proofs of Theorems~\ref{thm:polylog} and~\ref{thm:polylog2}. Let us mention that the proofs rely on a new concentration result for the area of the intersection of a random sphere cap with an arbitrary
subset of $\mathds S^{d-1}$, which is established using optimal transport maps and entropy-transport inequalities
on the unit sphere. \cite{Liu2021TestingTF} make use of this set-cap intersection concentration lemma for the theoretical analysis of the Belief Propagation algorithm.

\subsection{Open problems and perspectives}

The main results we have presented so far look as follows:

\begin{center}
\noindent
\renewcommand{\arraystretch}{1.5}
\begin{tabular}{|>{\centering\arraybackslash}m{.3\linewidth}|>{\centering\arraybackslash}m{.55\linewidth}|>{\centering\arraybackslash}m{.07\linewidth}|}
  \hline
{\bf Task} & {\bf Current state of knowledge} & {\bf Ref.} \\ \hline
Recognizing if a graph can be realized as a RGG & NP-hard & \citeyear{BREU19983}\\\hline
Testing between~$G(n,p,d)$ and~$G(n,p)$ in high-dimension for~$p\in(0,1)$ fixed &  
\resizebox{.5\textwidth}{!}{
  \makecell{\begin{tikzpicture}   
    \draw[->] (1,0) -- (9,0);
     \draw (5 cm,3pt) -- (5 cm,-3pt);
    \draw (1,0) node[above=3pt] {$ 0~$};
    \draw (5,0)  node[above=3pt] {$n^3$};
    \draw (3,0) node[above=3pt] {Polynomial time test} ;
    \draw (7,0) node[above=3pt] {Undistinguishable};
    \draw (9,0) node[above=3pt] {$ d~$};
  \end{tikzpicture}}
} & \citeyear{bubeck2016} \\ \hline
Testing between~$G(n,\frac{c}{n},d)$ and~$G(n,\frac{c}{n})$ in high-dimension for~$c>0$ &  
\resizebox{.5\textwidth}{!}{
  \makecell{\begin{tikzpicture}
    \draw[->] (1,0) -- (9,0);
     \draw (4 cm,3pt) -- (4 cm,-3pt);
     \draw (6 cm,3pt) -- (6 cm,-3pt);
    \draw (1,0) node[above=3pt] {$ 0~$};
    \draw (4,0)  node[above=3pt] {$\log^3 n$};
    \draw (6,0)  node[above=3pt] {$\log ^{36} n$};
    \draw (2.2,0) node[below=3pt] {Polynomial time test} ;
    \draw (7.5,0) node[below=3pt] {Undistinguishable};
        \draw (5,0) node[below=3pt] {?};
    \draw (9,0) node[above=3pt] {$ d~$};
  \end{tikzpicture}}}  & \citeyear{bubeck2016} \&  \citeyear{Liu2021TestingTF}\\ \hline
\end{tabular}
\end{center}

\begin{figure}[!ht]
\begin{minipage}{0.55\textwidth}
\includegraphics[scale=0.6]{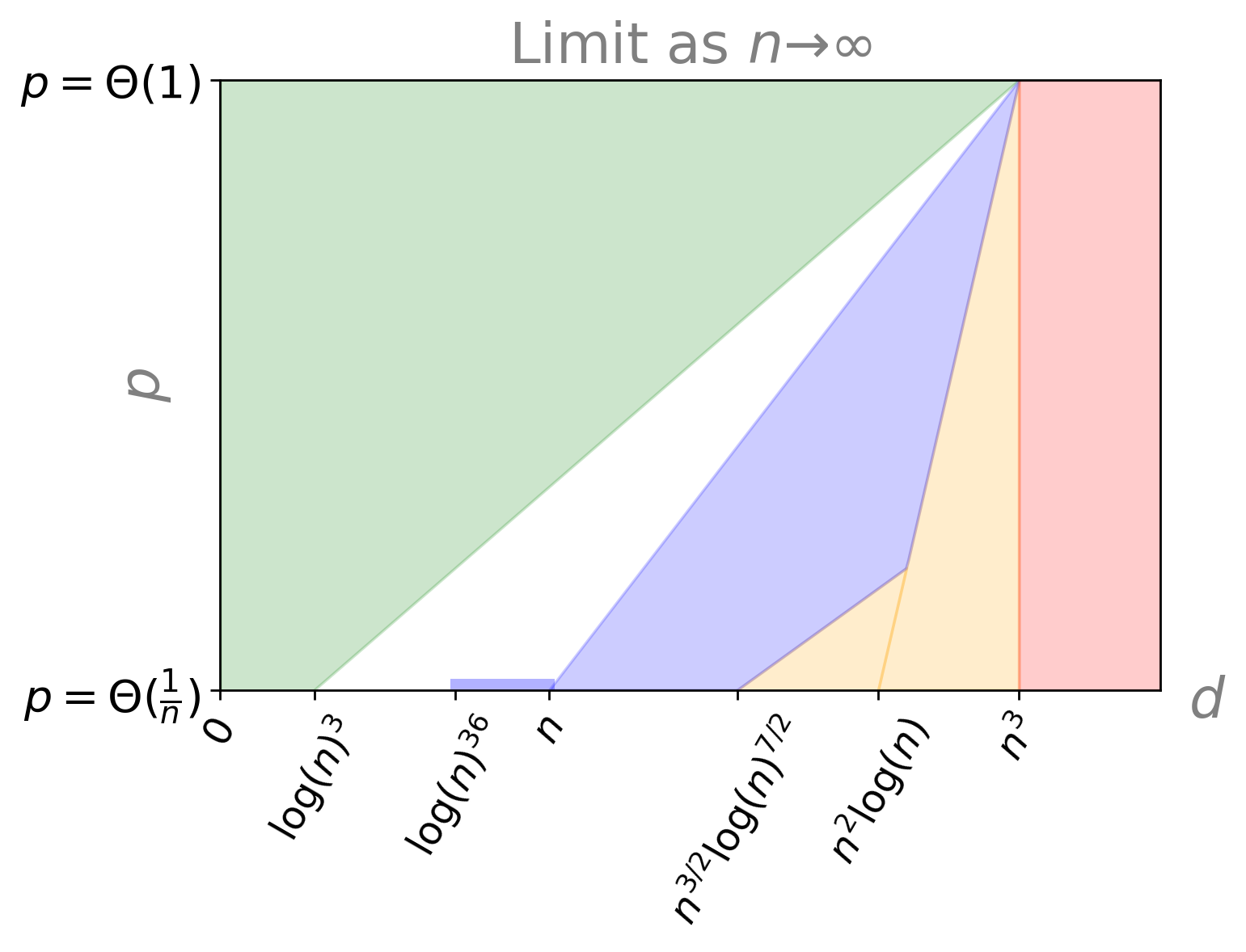}
\end{minipage}\hfill
\begin{minipage}{0.4\textwidth}
\begin{tabular}{cl}
\multicolumn{2}{c}{{\bf Phase where signed triangle counts differ}} \\
\multirow{2}{*}{\crule[green!80!blue!30!]{0.5cm}{0.3cm}} & \cite{bubeck2016}  \&\\
& \cite{Liu2021TestingTF}\\~\\
\multicolumn{2}{l}{{\bf Impossible Phase}} \\
\crule[red!30!white!100]{0.5cm}{0.3cm} & \cite{bubeck2016}\\
\crule[orange!30!white!100]{0.5cm}{0.3cm} & \cite{brennan20}\\
\crule[blue!30!white!100]{0.5cm}{0.3cm} & \cite{Liu2021TestingTF}
\end{tabular}
\end{minipage}
\caption{Phase-Diagram of the $(d,p)$ regions where geometry detection (on the Euclidean sphere) is known to be information theoretically impossible or possible (in polynomial time). Note that the figure only presents a simplified illustration of the current state of knowledge for the problem of geometry detection on $\mathds S^{d-1}$ since the true scales are not respected.}
\label{fig:phase_diagram}
\end{figure}
With new proof techniques based on combinatorial arguments,  direct couplings and  applications of  information  inequalities, \cite{brennan20} were the first to make a progress towards Conjecture~\ref{bubeck16-conjecture}. Nevertheless, their proof was heavily relying on a coupling step involving a De Finetti-type result that requires the dimension~$d$ to be larger than the number of points~$n$. \cite{Liu2021TestingTF} improved upon the previous bounds by polynomial factors with innovative proof arguments. In particular, their analysis makes use of the Belief Propagation algorithm and the cavity method, and relies on a new sharp estimate for the area of the intersection of a random sphere cap with an arbitrary subset of $\mathds S^{d-1}$. The proof of this new concentration result is an application of optimal transport maps and entropy-transport inequalities. Despite this recent progress, a large span of research directions remain open and we discuss some of them in the following.
\begin{enumerate}
\item {\it Closing the gaps for geometry detection on the Euclidean sphere $\mathds S^{d-1}$.}\\
Figure~\ref{fig:phase_diagram} shows that there are still important research directions to investigate to close the question of geometry detection regarding RGGs on $\mathds S^{d-1}$. First in the sparse regime, it would be desirable to finally know if Conjecture~\ref{bubeck16-conjecture} is true, meaning that the phase transition occurs when the latent dimension is of the order of $\log^3 n$. It could be fruitful to see if some steps in the approach from \cite{Liu2021TestingTF} could be sharpened in order to get down to the threshold $\log^3 n$. A question that seems even more challenging is to understand what happens in the regimes where $p=p(n) \in (\frac{1}{n},1)$ and $d =d(n)\in ( [H(p)n]^3,p^2n^3)$ (corresponding to the white region on Figure~\ref{fig:phase_diagram}). To tackle this question, one could try to extend the methods used in the sparse case by \cite{Liu2021TestingTF} to denser cases. Another possible approach to close this gap would be to dig deeper into the connections between the Wishart and GOE ensembles. One research direction to possibly improve the existing impossibility results regarding geometry detection would be to avoid the use of the data-processing inequality in Eq.\eqref{TV-Wishart-GOE} which makes us lose the fact that we do not observe the matrices~$W(n,d)$ and~$M(n)$ themselves. To some extent, we would like to take into account that some information is lost by observing only the adjacency matrices. In a recent work, \cite{brennan2021finettistyle} made a first step in this direction. They study the total variation distance between the Wishart and GOE ensembles when some given mask is applied beforehand. They proved that the combinatorial structure of the revealed entries, viewed as the adjacency matrix of a graph~$G$, drives the distance between the two distributions of interest. More precisely, they provide regimes for the latent dimension~$d$ based exclusively on the number of various small subgraphs in~$G$, for which the total variation distance goes to either~$0$ or~$1$ as~$n\to \infty$.
\item {\it How specific is the signed triangle statistic to RGGs?}\\
Let us mention that the signed triangle statistic has found applications beyond the scope of spatial networks. In \cite{jin19}, the authors study community based random graphs (namely the Degree Corrected Mixed Membership model) and are interested in testing whether a graph has only one community or multiple communities. They propose the Signed Polygon as a class of new tests. In that way, they extend the signed triangle statistic to~$m$-gon in the network for any~$m\geq 3$. Contrary to \cite{bubeck2016}, the average degree of each node is not known and the Degree Corrected Mixed Membership model allows degree heterogeneity. \qu{In \cite{jin19}, the authors} define the signal-to-noise ratio (SNR) using parameters of their model and they prove that a phase transition occurs, namely~$i)$ when the SNR goes to~$+\infty$, the Signed Polygon test is able to separate the alternative hypothesis from the null asymptotically, and~$ii)$ when the SNR goes to~$0$ (and additional mild conditions), then the alternative hypothesis is inseparable from the null.
\item {\it How the phase transition phenomenon in geometry detection evolves when other latent spaces are considered?}\\
This question is related to the robustness of the previous results with respect to the latent space.  Inspired by \cite{bubeck2016}, \cite{eldan20} provided a generalization of Theorem~\ref{bubeck16-thm1} considering an ellipsoid rather than the sphere~$\mathds S^{d-1}$ as latent space. They proved that the phase transition also occurs at~$n^3$ provided that we consider the appropriate notion of dimension which takes into account the anisotropy of the latent structure. \\ In \cite{dall02}, the clustering coefficient of RGGs with nodes uniformly distributed on the hypercube shows systematic deviations from the Erdos-Rényi prediction. 
\item {\it What is inherent to the connection function?}

Considering a fixed number of nodes, \cite{erba20} use a multivariate version of the central limit theorem to show that the joint probability of rescaled distances between nodes is normal-distributed as~$d\to \infty$. They provide a way to compute the correlation matrix. This work allows them to evaluate the average number of~$M$-cliques, i.e. of fully-connected subgraphs with~$M$ vertices, in high-dimensional RGGs and Soft-RGGs. They can prove that the infinite dimensional limit of the average number of~$M$-cliques in Erdös-Rényi graphs is the same of the one obtained from for Soft-RGGs with a continuous activation function. On the contrary, they show that for classical RGGs, the average number of cliques does not converge to the Erdös-Rényi prediction. This paper leads to think that the behavior of local observables in Soft-RGGs can heavily depend on the connection function considered. The work from \cite{erba20} is one of the first to address the emerging questions concerning the high-dimensional fluctuations of some statistics in RGGs. If they focused on the number of~$M$-cliques, one can also mention the recent work from \cite{grygierek20} that provide a central limit theorem for the edge counting statistic as the space dimension~$d$ tends to infinity. Their work shows that the Malliavin–Stein approach for Poisson functionals that was first introduced in stochastic geometry can also be used to deal with spatial random models in high dimensions.
\medskip

In a recent work, \cite{liu21} are interested in extending the previous mentioned results on geometry detection in RGGs to Soft RGGs with some specific connection functions. The authors consider the dense case where the average degree of each node scales with the size of the graph~$n$ and study geometry detection with graphs sampled from Soft-RGGs that interpolate between the standard RGG on the sphere~$\mathds S^{d-1}$ and the Erdös-Rényi random graph. Hence, the null hypothesis remains that the observed graph~$G$ is a sample from~$G(n,p)$ while the alternative becomes that the graph is the Soft-RGG where we draw an edge between nodes~$i$ and~$j$ with probability 
\[  (1 - q)p + q \mathds 1_{t_{p,d} \leq \langle X_i , X_j\rangle},\] where~$(X_i)_{i\geq1}$ are randomly and independently sampled on~$\mathds S^{d-1}$ and where~$q \in [0,1]$ can be interpreted as the geometric strength of the model. Denoting the random graph model~$G(n,p,d,q)$, one can easily notice that~$G(n,p,d,1)$ is the standard RGG on the Euclidean sphere~$\mathds S^{d-1}$ while~$G(n,p,d,0)$ reduces to the Erdös-Rényi random graph. Hence, by taking~$q=1$ in Theorem~\ref{liu21-thm1}, we recover Theorem~\ref{bubeck16-thm1} from \cite{bubeck2016}. One can further notice that Theorem~\ref{liu21-thm1} depicts a polynomial
dependency on~$q$ for geometry detection but when~$q<1$ there is a gap between the upper and lower bounds as illustrated by Figure~\ref{fig:phase-diagram} taken from \cite{liu21}. As stated in \cite{liu21}, \textit{"[...] a natural direction of future research is to consider [geometry detection] for other
connection functions or underlying latent spaces, in order to understand how the dimension threshold
for losing geometry depends on them."}

\begin{figure}[!ht]
\centering
\begin{minipage}[c]{0.53\linewidth}
\begin{theorem} \label{liu21-thm1} \cite[Theorem 1.1]{liu21}

\noindent Let~$p\in (0,1)$ be fixed.
\begin{enumerate}[label=(\roman*)]
\item If~$n^3q^6/d\to \infty$, then  \[ \mathrm{TV}(G(n, p), G(n, p, d, q))\to 1\quad \text{as}~n\to \infty.\]
\item If~$nq \to 0$ or~$n^3q^2/d \to 0$, then \[\mathrm{TV} (G(n, p), G(n, p, d,q))\to 0\quad \text{as}~n\to \infty.\]
\end{enumerate}
\end{theorem}
\end{minipage} 
\begin{minipage}[c]{0.36\linewidth}
\centering
\includegraphics[scale=0.35]{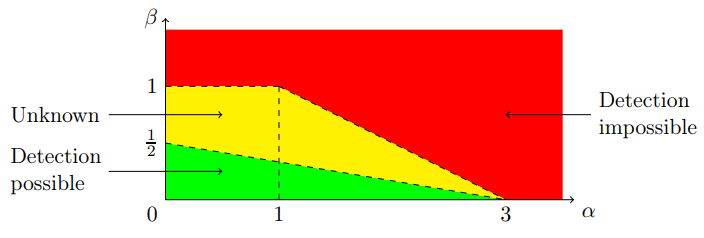}
\caption{Phase diagram for detecting geometry in the soft random geometric graph~$G(n, p, d, q)$. Here~$d = n^{\alpha}$ and~$q = n^{-\beta}$ for some~$\alpha, \beta > 0$.}
\label{fig:phase-diagram}
\end{minipage}
\end{figure}

The same authors in \cite{liuF21} extend the model of the Soft-RGG by considering the latent space $\mathds R^d$ where the latent positions~$(X_i)_{i\in [n]}$ are i.i.d. sampled with~$X_1 \sim \mathcal N(0,I_d)$. Two different nodes $i,j\in [n]$ are connected with probability $\mathbf p(\langle X_i,X_j\rangle)$ where $\mathbf p$ is a monotone increasing connection function. More precisely, they consider a connection function $\mathbf p$ parametrized by $i)$ a cumulative distribution function $F:\mathds R\to [0,1]$ and $ii)$ a scalar $r>0$ and given by \[\mathbf p :t \mapsto F\left( \frac{t-\mu_{p,d,r}}{r\sqrt d} \right),\]
where $\mu_{p,d,r}$ is determined by setting the edge density in the graph to be equal to $p$, namely $\mathds E\left[ \mathbf p (\langle X_1,X_2\rangle)\right] =p.$ They work in the dense regime by considering that $p\in(0,1)$ is independent of $n$.
The parameter $r$ encodes the flatness of the connection function and is typically a function of $n$. The authors prove phase transitions of detecting geometry in this framework in terms of the dimension of the underlying geometric space $d$ and the variance parameter $r$. The larger $r$, the smaller the dimension $d$ at which the phase transition occurs. When $r \underset{n \to \infty}{\to}0$, the connection function becomes an indicator function and the transition appears at $d\asymp n^3$ (recovering the result from Theorem~\ref{bubeck16-thm1} established for RGGs on the Euclidean Sphere).

\item {\it Suppose that we know that the latent variables are embedded in a Eucliden Sphere, can we estimate the dimension~$d$ from the observation of the graph?}\\
When~$p=1/2$, \cite{bubeck2016} obtained a bound on the difference of the expected number of signed triangles between consecutive dimensions leading to Theorem~\ref{bubeck16-thm5}.
 
\begin{theorem} \label{bubeck16-thm5} \cite[Theorem 5]{bubeck2016}

\noindent There exists a universal constant~$C >0$, such that for all integers~$n$ and~$d_1< d_2$, one has
 \[\mathrm{TV} (G(n,1/2, d_1), G(n,1/2, d_2))\geq 1-C\left(\frac{d_1}{n}\right)^2.\]
\end{theorem}
The bound provided by Theorem~\ref{bubeck16-thm5} is tight in the sense that when~$d \gg n$,~$G(n,1/2,d)$ and~$G(n,1/2,d+1)$ are indistiguishable as proved in \cite{eldan15}.
More recently, \cite{araya19} proposed a method to infer the latent dimension of a Soft-RGG on the Euclidean Sphere in the low dimensional setting. Their approach is proved to correctly recover the dimension~$d$ in the relatively sparse regime as soon as the connection function belongs to some Sobolev class and satisfies a spectral gap condition.
\item  {\it Extension to hypergraphs and information-theoretic/computational gaps.}\\
Let us recall that a hypergraph is a generalization of a graph in which an edge can join any number of vertices. Extensions of RGGs to hypergraphs have already been proposed in the literature (see for example \cite{hypergraphs}). A nice research direction would consist in investigating the problem of geometry detection in these geometric representations of random hypergraphs. As already discussed, it has been conjectured that the problem of geometry detection in RGGs on $\mathds S^{d-1}$ does not present a statistical-to-algorithmic gap meaning that whenever it is information theoretically possible to differ $G(n,p,d)$ from $G(n,p)$, we can do it with a computational complexity polynomial in $n$ (using the signed triangle statistic). Dealing with hypergraphs, one can legitimately think that statistical-to-algorithmic gaps could emerge. This intuition is based on the fact that most of the time, going from a matrix problem to a tensor problem brings extra challenges. One can take the example of principal component analysis of Gaussian $k$-tensors with a planted rank-one spike (cf. \cite{arous20}). In this problem, we assume that we observe for any $l\in [n]$,
\[\mathbf Y^l=\lambda u^{\otimes k}+\mathbf W^l,\]
where $u \in \mathds S^{d-1}$ is deterministic, $\lambda \geq 0$ is the signal-to-noise ratio and where $(\mathbf W^l)_{l\in [n]}$ are independent Gaussian $k$-tensor (we refer to \cite{arous20} for further details). The goal is to infer the “planted signal” or “spike”, $u$. In the matrix case (i.e. when $k=2$), whenever the problem is information theoretically solvable, we can also recover the spike with a polynomial time algorithm (using for example a spectral method). If we look at the tensor version of this problem where $k\geq 3$, there is a regime of signal-to-noise
ratios for which it is information theoretically possible to recover the signal but for which there
is no known algorithm to approximate it in polynomial time in $n$. This is a statistical-to-algorithmic gap and we refer to \cite[Section 3.8]{brennan2020reducibility} and references therein for more details.
\item {\it Can we describe the properties of high dimensional RGGs in the regimes where $\mathrm{TV}(G(n,p),G(n,p,d)) \to 1$ as $n\to \infty$?}\\
In the low dimensional case, RGGs have been extensively studied: their spectral or topological properties, chromatic number or clustering number are now well known (see e.g. \cite{walters11,penrose03}). One of the first work studying the properties of high dimensional RGGs is \cite{avrachenkov20} where the authors are focused on the clique structure. These questions are essential to understand how good high dimensional RGGs are as models for the theory of network science. 
\item {\it How to find a relevant latent space given a graph with an underlying geometric structure?}\\
As stated in \cite{racz16}, {\it "Perhaps the ultimate goal is to find good representations of network data, and hence to faithfully embed the graph of interest into an appropriate metric space"}. This task is known as {\it manifold learning} in the Machine learning community. Recently \cite{smith19} proved empirically that the eigenstructure of the Laplacian of the graph provides information on the curvature of the latent space. This is an interesting research direction to propose model selection procedure and infer a relevant latent space for a graph.
\end{enumerate}

\section{Non-parametric inference in RGGs}
\label{sec:non-parametric}

In this section, we are interested in non-parametric inference in TIRGGs (see Definition~\ref{def:tirgg}) on the Euclidean sphere~$\mathds S^{d-1}$. The methods presented rely mainly on spectral properties of such random graphs. Note that spectral aspects in (Soft-)RGGs have been investigated for a long time (see for example \cite{rai04}) and it is now well-known that the spectra of RGGs are very different from the one of other random graph models since the appearance of particular subgraphs  give  rise  to  multiple  repeated  eigenvalues (see \cite{nyberg15} and \cite{blackwell07}). Recent works took advantage of the information captured by the spectrum of RGGs to study topological properties such as \cite{aguilar20}. In this section, we will see that random matrix theory is a powerful and convenient tool to study the spectral properties of RGGs as already highlighted by \cite{Dettmann17}.

\subsection{Description of the model and notations}
\label{sec:graphon}

We consider a Soft-RGG on the Euclidean Sphere~$\mathds S^{d-1}$ endowed with the geodesic distance~$\rho$. We consider that the connection function~$H$ is of the form~$H:t \mapsto \mathbf p(\cos(t))$ where~$\mathbf p:[-1,1] \to [0,1]$ is an unknown function that we want to estimate. This Soft-RGG belongs to the class of TIRGG has defined in Section~\ref{sec:models} and corresponds to a graphon model where the graphon~$W$ is given by \[\forall x,y \in \mathds S^{d-1}, \quad W(x,y) :=\mathbf p(\langle x,y\rangle).\] 
$W$ viewed as an integral operator on square-integrable functions, is a compact convolution (on the left) operator
\begin{equation}\mathds T_W:f \in L^2(\mathds S^{d-1}) \mapsto \int_{\mathds S^{d-1}} W(x,\cdot)f(x)\sigma(dx) \in L^2(\mathds S^{d-1}), \label{eq:inte-ope}\end{equation} where~$\sigma$ is the Haar measure on~$\mathds S^{d-1}$. The operator~$\mathds T_W$ is Hilbert-Schmidt and it has a countable
number of bounded real eigenvalues~$\lambda^*_k$ with zero as only accumulation point. The eigenfunctions
of~$\mathds T_W$ have the remarkable property that they do not depend on~$p$ (see \cite[Lemma 1.2.3]{xu13}): they are given by the real Spherical Harmonics. We denote~$\mathcal H_l$ the space of real Spherical Harmonics of degree~$l$ with dimension~$d_l$ and with orthonormal basis~$(Y_{l,j})_{j\in [d_l]}$. We end up with the following spectral decomposition of the {\it envelope} function~$\mathbf p$
\begin{equation}\forall x,y \in \mathds S^{d-1}, \quad \mathbf p(\langle x, y\rangle) = \sum_{l\geq 0}p^*_l \sum_{j=1}^{d_l}Y_{l,j}(x)Y_{l,j}(y) =\sum_{l\geq 0}p^*_l c_l G_l^{\beta}(\langle x,y \rangle), \label{decompo-p}\end{equation}
where~$\lambda^* := (p^*_0, p^*_1,\dots, p^*_1,\dots, p^*_l,\dots , p^*_l,\dots)$ meaning that each eigenvalue~$p^*_l$ has multiplicity~$d_l$ and~$G_l^{\beta}$ is the Gegenbauer polynomial of degree~$l$ with parameter~$\beta := \frac{d-2}{2}$ and
$c_l := \frac{2l+d-2}{d-2}$.~$\mathbf p$ is assumed bounded and as a consequence~$\mathbf p\in L^2((-1,1),w_{\beta})$ where the weight function~$w_{\beta}$ is defined by~$w_{\beta}(t) := (1-t^2)^{\beta-1/2}$. Note that the decomposition~\eqref{decompo-p} shows that it is enough to estimate the eigenvalues~$(p^*_l)_l$ to recover the envelope function~$\mathbf p$. 

\subsection{Estimating the matrix of probabilities}
Let us denote~$A$ the adjacency matrix of the Soft-RGG~$G$ given by entries~$A_{i,j}\in \{0,1\}$ where~$A_{i,j}= 1$ if the nodes~$i$ and~$j$ are connected and~$A_{i,j}= 0$ otherwise. We denote by~$\Theta$ the~$n\times n$ symmetric matrix with entries~$\Theta_{i,j}=\mathbf p\left(\langle X_i,X_j\rangle\right)$ for~$1\leq i < j\leq n$ and zero diagonal entries. We consider the scaled version of the matrices~$A$ and~$\Theta$ given by
\[\widehat T_n = \frac{1}{n}A\quad \text{and}  \quad T_n = \frac{1}{n}\Theta.\] \cite{bandeira16} proved a near optimal error bound for the operator norm of~$\widehat T_n-T_n$. Coupling this result with the Weyl’s perturbation Theorem gives a control on the difference between the eigenvalues of the matrices~$\widehat T_n$ and~$T_n$, namely with probability greater that~$1-\exp(-n)$ it holds,\begin{equation}\label{bandeira:iid-RGG}\forall k\in [n],\quad |\lambda_k(\widehat T_n)-\lambda_k(T_n)|\leq \|\widehat T_n-T_n\|=O(1/\sqrt n),\end{equation}where~$\lambda_k(M)$ is the~$k$-th largest eigenvalue of any symmetric matrix~$M$. This result shows that the spectrum of the scaled adjacency matrix~$\widehat T_n$ is a good approximation of the one of the scaled matrix of probabilities~$T_n$.

\subsection{Spectrum consistency of the matrix of probabilities}
For any~$R \geq 0$, we denote \begin{equation}\label{notation-tilde}\tilde R:= \sum_{l=0}^R d_l,\end{equation} which corresponds to the dimension of the space of Spherical Harmonics with degree at most~$R$. Proposition~\ref{prop4:iid-RGG} states that the spectrum of~$T_n$ converges towards the one of the integral operator~$\mathds T_W$ in the~$\delta_2$ metric which is defined as follows. 
\begin{definition}
Given two sequences~$x,y$ of reals--completing finite sequences by zeros--such that~$\sum_i x_i^2+y_i^2<\infty$, we define the~$\ell_2$ rearrangement distance~$\delta_2(x,y)$ as 
\[\delta_2^2(x,y):= \inf_{\sigma \in \mathfrak{S}_n}\sum_i (x_i-y_{\sigma(i)})^2\,,
\]
where~$\mathfrak{S}_n$ is the set of permutations with finite support. This distance is useful to compare two spectra.
\end{definition}
\begin{proposition} \label{prop4:iid-RGG} \cite[Proposition 4]{decastro20}

\noindent There exists a universal constant~$C >0$ such that for all~$\alpha\in (0,1/3)$ and for all~$n^3\geq  \tilde R \log(2\tilde R/\alpha)$, it holds \begin{equation}\label{eq:prop4-iid-RGG} \delta_2(\lambda(T_n),\lambda^*)\leq 2\left[ \sum_{l>R}d_l\left(p^*_l\right)^2\right]^{1/2}+C \sqrt{\tilde R\left(1+\log(\tilde R/\alpha)\right)/n},\end{equation}
with probability at least~$1-3\alpha$. \end{proposition}
Proposition~\ref{prop4:iid-RGG} shows that the~$\ell_2$ rearrangement distance between~$\lambda^*$ and~$\lambda(T_n)$ decomposes as the sum of a bias term and a variance term. The second term on the right hand side of~\eqref{eq:prop4-iid-RGG} corresponds to the variance. The proof leading to this variance bound relies on the Hoffman-Wielandt inequality and borrows ideas from \cite{Gine}. It makes use of recent developments in random matrix concentration by applying a Bernstein-type concentration inequality (see \cite{tropp15} for example) to control the operator norm of the sum of independent centered symmetric matrices given by
\begin{equation}\label{eq:sum-mat-iid}\sum_{i=1}^n \left(\mathbf Y(X_i)\mathbf Y(X_i)^{\top} - \mathds E\left[\mathbf Y(X_i)\mathbf Y(X_i)^{\top}\right] \right) ,\end{equation}
with~$ \mathbf Y(x) = \left(Y_{0,0}(x),Y_{1,1}(x), \dots, Y_{1,d_1}(x),Y_{2,1}(x), \dots, Y_{2,d_2}(x), \dots, Y_{R,1}(x), \dots, Y_{R,d_R}(x) \right)^{\top} \in \mathds R^{\tilde R}$ for all~$x \in \mathds S^{d-1}$. The proof of Proposition~\ref{prop4:iid-RGG} also exploits concentration inequality for U-statistic dealing with a bounded, symmetric and~$\sigma$-canonical kernel (see \cite[Definition 3.5.1]{de2012decoupling}). The first term on the right hand side of~\eqref{eq:prop4-iid-RGG} is the bias arising from choosing a resolution level equal to~$R$. Its behaviour as a function of~$R$ can be analyzed by considering some regularity condition on the envelope~$\mathbf p$. Assuming that~$\mathbf p$ belongs to the Sobolev class~$Z^s_{w_{\beta}}((-1,1))$ (with regularity encoded by some parameter $s>0$) defined by
 \[ \left\{g=\sum_{k \geq0} g^*_k c_k G_k^{\beta}\in L^2((-1,1),w_{\beta})\; \Bigg|\;\| g\|_{Z^s_{w_{\beta}}((-1,1))}^*:= \left[ \sum_{l=0}^{\infty}d_l |g_l^*|^2\left(1+ (l(l+2\beta))^s\right) \right]^{1/2}<\infty \right\},\] and choosing the resolution level~$R_{opt}=\ceil{(n/\log n)^{\frac{1}{2s+d-1}}}$ to balance the bias/variance tradeoff appearing on the right hand side of~\eqref{eq:prop4-iid-RGG}, we get that
\begin{equation*}\mathds E \left[ \delta_2^2\left(   \lambda(T_n),\lambda^{*}\right)\right] \lesssim\left[ \frac{n}{\log n}\right]^{-\frac{2s}{2s+(d-1)}} .\end{equation*}
Thus we recover a classical nonparametric rate of convergence for estimating a function with smoothness~$s$ in a space of dimension~$d-1$. This is also the rate towards the probability matrix obtained by \cite{xu18}. Note that the choice of~$R_{opt}$ requires the knowledge of the regularity parameter~$s$. To overcome this issue, \cite{decastro20} proposed an adaptive procedure using the Goldenshluger-Lepski method.

\subsection{Estimation of the envelope function}
\label{sec:esti-env}

Let us denote~$\lambda:= \lambda(\widehat T_n)$. For a prescribed model size~$R \in \mathds N\backslash\{0\}$, \cite{decastro20} define the estimator~$\widehat \lambda^R$ of the truncated spectrum~$\lambda^{*R}:= (p_0^*, p^*_1, \dots, p_1^*, \dots,p^*_R, \dots,p^*_R)$ of~$\lambda^*$ as
\[\widehat \lambda^R:=( p_0^R(\hat \sigma), p_1^R(\hat \sigma),\dots, p_1^R(\hat \sigma),\dots, p_1^R(\hat \sigma),\dots, p_R^R(\hat \sigma),\dots, p_R^R(\hat \sigma)),\]
with \[\hat \sigma \in \underset{\sigma \in \mathfrak{S}_n }{\arg \min}  \sum_{l=0}^R \sum_{k=\widetilde{l-1}}^{\widetilde l}\left(  p_l^R(\sigma)-\lambda_{\sigma(k)}\right)^2   +\sum_{k=\tilde R+1}^n \lambda_{\sigma(k)}^2 \quad \text{and} \quad p_l^R(\sigma) = \frac{1}{d_l} \sum_{k=\widetilde{l-1}}^{\widetilde{l}} \lambda_{\sigma(k)},\]
where~$\mathfrak{S}_n$ is the set of permutations of~$[n]$ and where we used the notation~\eqref{notation-tilde} with the convention~$\widetilde{-1}=1$.
Using the results of the two previous subsections namely~\eqref{bandeira:iid-RGG} and Proposition~\ref{prop4:iid-RGG}, we obtain \cite[Theorem.6]{decastro20} which states that \begin{equation*}\mathds E \left[ \delta_2^2\left(  \widehat{\lambda}^{R_{opt}},\lambda^{*}\right)\right] \lesssim\left[ \frac{n}{\log n}\right]^{-\frac{2s}{2s+(d-1)}} .\end{equation*}
The envelope function~$\mathbf p$ can then be approximated by the plug-in estimator~$\widehat{\mathbf{  p}} \equiv \sum_{l=0}^{R_{opt}}p^{R_{opt}}_l(\hat{\sigma}) c_l G_l^{\beta}$ based on the decomposition~\eqref{decompo-p}. One drawback of this approach is the exponential complexity in~$R$ of the computation of~$\widehat\lambda^R$. In the next section, we will describe an approach based on a Hierarchical Agglomerative Clustering algorithm to estimate the envelope function~$\mathbf p$ efficiently.

\subsection{Open problems and perspectives}

The minimax rate of estimating a~$s$-regular function on a space of (Riemannian) dimension~$d-1$ such as~$\mathds S^{d-1}$ from~$n$ observations is known to be of order~$n^{-\frac{s}{2s+d-1}}$. In the framework of this section, even if the domain of the envelope function~$\mathbf p$ is~$[-1,1]$, inputs of~$\mathbf p$ are the pairwise distances given by inner products of points embedded in~$\mathds S^{d-1}$. Hence it is still an open question to know if the dimension~$d$ of the latent space appears in the minimax rate of convergence. Moreover, the number of observations in the estimation problem considered is~$n^2$ since the full adjacency matrix is known. Nevertheless the problem suffers from the presence of unobserved latent variables. This all contributes to a non standard estimation problem and  finding the optimal rate of convergence is an open problem.

\section{Growth-model in RGGs}
\label{sec:MRGG}

\subsection{Description of the model}
\label{model-MRGG}
In \cite{decastroduchemin20}, a new growth model was introduced for RGGs. The so-called Markov Random Geometric Graph (MRGG) already presented in Definition~\ref{def:mrgg} is a Soft-RGG where latent points are sampled with Markovian jumps. Namely, \cite{decastroduchemin20} consider~$n$ points~$X_1, X_2, \dots, X_n$ sampled on the Euclidean sphere~$\mathds{S}^{d-1}$ using a Markovian dynamic. They start by sampling uniformly~$X_1$ on~$\mathds{S}^{d-1}$. Then, for any~$i \in \{2, \dots,n\}$, they sample
\noindent
\begin{figure}[!ht]
\begin{minipage}{0.65\linewidth}
\begin{itemize}
\item a unit vector~$Y_i \in \mathds{S}^{d-1}$ uniformly, orthogonal to~$X_{i-1}$,
\item a real~$r_i \in [-1,1]$ encoding the distance between~$X_{i-1}$ and~$X_i$, see~\eqref{e:geo}.~$r_i$ is sampled from a distribution~$f_{\mathcal{L}}:[-1,1] \to [0,1]$, called the \textit{latitude function}, 
\end{itemize}
then~$X_i$ is defined by 
\begin{equation}
\label{e:jump}
X_i = r_i \times X_{i-1} + \sqrt{1-r_i^2}\times  Y_i\,.
\end{equation}
\end{minipage}\hfill
\begin{minipage}{0.3\linewidth}
\centering
\includegraphics[scale=0.24]{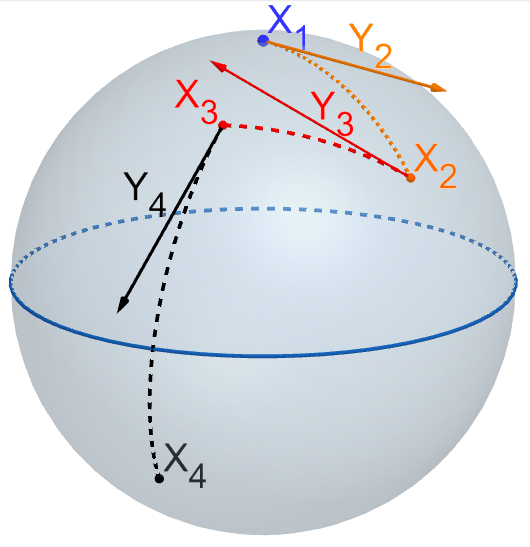}
\caption{Visualization of the sampling scheme in~$\mathds S^{2}$.}
\label{fig:MRGG-smampling}
\end{minipage}
\end{figure}

This dynamic is illustrated with Figure~\ref{fig:MRGG-smampling} and can be understood as follows. Consider that~$X_{i-1}$ is the north pole, then choose uniformly a direction (i.e. a longitude) and, in an independent manner, randomly move along the latitudes (the longitude being fixed by the previous step). The geodesic distance~$\gamma_i$ drawn on the latitudes satisfies 
\begin{equation}
\label{e:geo}
\gamma_i=\arccos(r_i)\,,
\end{equation}
where random variable~$r_i=\langle X_i, X_{i-1}\rangle$ has density~$f_{\mathcal{L}}(r_i)$.

\subsection{Spectral convergences}
In this framework and keeping the notations of the previous section, one can show that if~$\mathbf p\in Z^s_{w_{\beta}}((-1,1))$ and if~$f_{\mathcal L}$ satisfies the condition
\[(\mathcal{H}) \qquad \|f_{\mathcal L}\|_{\infty}:=\sup_{t\in[-1,1]} |f_{\mathcal L}(t)|<\infty \quad \text{and} \quad f_{\mathcal L}\;\; \text{is bounded away from zero},\]then \begin{equation}\label{eq:rate-markov}\mathds E\left[\delta_2^2(\lambda(T_n),\lambda^*)\vee \delta_2^2(\lambda^{R_{opt}}(\widehat{T}_n),\lambda^*)\right]=\mathcal{O}\left(\left[ \frac{n}{\log^2(n)} \right]^{- \frac{2s}{2s+d-1}}\right),\end{equation}
with $\lambda^{R_{opt}}(\hat{T}_n)=(\hat{\lambda}_1, \dots,\hat{\lambda}_{\tilde R_{opt}}, 0,0,\dots)$ and $R_{opt}=\floor{\left(n/\log^2(n)\right)^{\frac{1}{2s+d-1}}}$ where $\hat{\lambda}_1, \dots,\hat{\lambda}_n$ are the eigenvalues of $\widehat{T}_n$ sorted in decreasing order of magnitude. 
This result is the counterpart of Proposition~\ref{prop4:iid-RGG} in this Markovian framework. The proof follows closely the steps of the one of the previous section but one needs to deal with the dependency of the latent positions. Results from \cite{tropp15} are no longer suited to control the operator norm of~\eqref{eq:sum-mat-iid} since~$(X_i)_{i\geq 0}$ is a Markov chain. Nevertheless, this can be achieved by using concentration inequalities for sum of functions of Markov chains and by exploiting the rank one structure of the random matrices~$\mathbf Y(X_i)\mathbf Y(X_i)^{\top}$ together with a covering set argument. Another difficulty induced by the latent dynamic is the control of a U-statistic of order 2 of the Markov chain~$(X_i)_{i\geq 0}$ with a bounded kernel. Non-asymptotic results regarding the tail behaviour of U-statistics of a Markov chain have been so far very little touched. In a recent work, \cite{duchemin20} proved a concentration inequality for order 2 U-statistics with bounded kernels for uniformly ergodic Markov chain. Theorem~\ref{thm-markov-Ustat} gives a simplified version of their main result. 
Assuming that the condition~$(\mathcal H)$ is fulfilled, the Markov chain~$(X_i)_{i\geq 1}$ satisfies the assumptions of Theorem~\ref{thm-markov-Ustat} and one can show that~\eqref{eq:rate-markov} holds true.
\begin{theorem} \cite[Theorem 2]{duchemin20} \label{thm-markov-Ustat} 
\noindent Let us consider a Markov chain~$(X_i)_{i\geq 1}$ on some measurable space~$(E,\mathcal E)$ (with $E$ Polish) with transition kernel~$P:E\times E \to \mathds R$ and a function~$h: E \times E \to \mathds R$. We assume that
\begin{enumerate}
\item~$(X_i)_{i\geq 1}$ is a uniformly ergodic Markov chain with invariant distribution~$\pi$,
\item~$h$ is bounded and~$\pi$-{\it canonical}, namely \[ \forall x \in E, \quad \mathds E_{X \sim \pi}[h(X,x)]=\mathds E_{X \sim \pi}[h(x,X)]=0,\]
\item there exist~$\delta>0$ and some probability measure~$\nu$ on~$(E,\mathcal E)$ such that
\[\forall x \in E,\; \forall A \in \mathcal E, \quad P(x,A) \leq \delta \nu(A).\]
\end{enumerate}
Then there exist constants~$\beta,\kappa>0$ such that for any~$u \geq 1$, it holds with probability at least~$1-\beta e^{-u}\log n$,
\[\frac{1}{n(n-1)} \sum_{1\leq i, j\leq n, \; i\neq j} h(X_i,X_j) \leq \kappa \|h\|_{\infty}\log n\,\, \bigg\{   \frac{u}{n} + \left[\frac{u}{n}\right]^{2}  \bigg\}, \]
where~$\kappa$ and~$\beta$ only depend on constants related to the Markov chain~$(X_i)_{i \geq 1}$.
\end{theorem}
{\bf Remark.} Note that Theorem~\ref{thm-markov-Ustat} holds for any initial distribution of the Markov chain. In their paper, \cite{duchemin20} go beyond the previous Hoeffding tail control by providing a Bernstein-type concentration inequality under the additional assumption that the chain is stationary. For the sake of simplicity we presented Theorem~\ref{thm-markov-Ustat} for a single kernel~$h$, but we point out that their results allow for the dependence of the kernels -- say~$h_{i,j}$ -- on the indexes in the sums which brings technical difficulties since standard blocking techniques can no longer be applied. The interest for this concentration result goes beyond the scope of random graphs since U-statistics naturally arise in online learning~\cite{clemencon08} or testing procedures~\cite{fromont2006}.

\subsection{Estimation procedure}Recalling the notation of the truncated spectrum $\lambda^{*R}$ (resp. $\lambda^{R}(\widehat{T}_n)$) of $\lambda^*$ (resp. $\lambda(\widehat{T}_n)$) from Section~\ref{sec:esti-env}, \cite{decastroduchemin20} introduce a new procedure (namely the SCCHEi algorithm) based on a Hierarchical Agglomerative Clustering that returns a partition $\mathcal C_{d_0},\dots,\mathcal C_{d_R},\Lambda$ of the $n$ eigenvalues of $\widehat T_n$ where for any $i \in \{0,\dots,R\}$, $|\mathcal C_{d_i}|=d_i$ (where we recall that $d_i$ is the dimension of the space of spherical Harmonics of degree $i$). The authors prove that for any fixed resolution level $R$, $n$ can be chosen large enough so that the clusters obtained in polynomial time from the SCCHEi algorithm satisfy
\begin{equation}
\label{eq:thm-SCCHEi}\delta_2^2(\lambda^{*R},\lambda^{R}(\widehat{T}_n)) = \sum_{k=0}^R \sum_{\hat \lambda\in \mathcal C_{d_k}} (\hat \lambda - p^*_k)^2.
\end{equation}
The final estimate of the envelope function with resolution level~$R$ is defined as \begin{equation} \label{env-hat}\widehat {\mathbf p}:= \sum_{k=0}^R \widehat p_k G_k^{\beta}, \quad \text{where} \quad\forall k \in \mathds{N},\quad \widehat{p}_k = \left\{
    \begin{array}{ll}
        \frac{1}{d_k} \sum_{ \lambda \in \mathcal{C}_{d_k}} \lambda  & \mbox{if } k \in \{0,\dots,R \}\\
        0 & \mbox{otherwise.}
    \end{array}
\right.  \end{equation}Eq.\eqref{eq:thm-SCCHEi} is not a sufficient condition to ensure that the $L^2$ error between the true envelope function and the plug-in estimator $\widehat {\mathbf p}$ (see Eq.\eqref{env-hat}) goes to $0$ has $n \to +\infty.$ This is due to identifiability issues coming from the $\delta_2$ metric. In \cite[Theorem 3]{decastroduchemin20}, the author obtain a theoretical guarantee on the $L^2$ error between the true envelope function and the plug-in estimate by considering additional assumptions on the eigenvalues $(p^*_k)_{k\geq0}$. Let us finally mention that the optimal resolution level $R_{opt}$ is unknown in practice. To bypass this issue, the authors propose a model selection procedure based on the slope heuristic (see \cite{arlot2019}).

\subsection{Non-parametric link prediction}

We are now interesting in solving link prediction tasks. Namely, from the observation of the graph at time~$n$, we want to estimate the probabilities of connection between the upcoming node~$n+1$ and the nodes already present in the graph. Recalling the definition of the random variables $(Y_i)_{i\geq 2}$ from Section~\ref{model-MRGG} and denoting further~$ \mathrm{proj}_{X_n^{\perp}}(\cdot)$ the orthogonal projection onto the orthogonal complement of~$\mathrm{Span(X_n)}$, the decomposition 
 \begin{equation}\label{eq:decompo}\langle X_i,X_{n+1} \rangle=\langle X_i,X_n \rangle \langle X_n , X_{n+1} \rangle + \sqrt{1-\langle X_n , X_{n+1} \rangle^2} \sqrt{1-\langle X_i,X_n \rangle^2} \langle \frac{\mathrm{proj}_{X_n^{\perp}}(X_i)}{\|\mathrm{proj}_{X_n^{\perp}}(X_i)\|_2}, Y_{n+1}\rangle ,\end{equation} shows that latent distances~$\mathbf D_{1:n} = (\langle X_i,X_j\rangle )_{1\leq i,j\leq n} \in [-1,1]^{n\times n}$ are enough for link prediction. Indeed, it can be achieved by estimating the posterior probabilities defined for any $i \in [n]$ by
\begin{align}
\eta_i(\mathbf D_{1:n})& = \mathds P\left( A_{i,n+1}=1 \;|\; \mathbf D_{1:n}  \right) \notag \\
\eta_i(\mathbf D_{1:n})&=\int_{r,u\in (-1,1)} \mathbf p\left(   \langle X_i,X_n \rangle r + \sqrt{1-r^2} \sqrt{1-\langle X_i,X_n \rangle^2}u\right) f_{\mathcal{L}}(r) w_{\frac{d-3}{2}}(u)\frac{\Gamma(\frac{d-1}{2})}{\Gamma(\frac{d-2}{2})\sqrt \pi}drdu,
 \label{eta_i}
\end{align}
where~$A_{i,n+1} \in \{0,1\}$ is one if and only if node~$n+1$ is connected to node~$i$, $w_{\frac{d-3}{2}}(u):=(1-u^2)^{\frac{d-3}{2}-\frac12}$ and where $\Gamma:a \in ]0,+\infty[\mapsto \int_0^{+\infty} t^{a-1}e^{-t}dt$. Using an approach similar to \cite{araya19}, \cite{decastroduchemin20} proved that one can get a consistent estimator~$\widehat G$ of the Gram matrix of the latent positions~$G=\left( \langle X_i,X_j \rangle \right)_{1\leq i,j\leq n}$ in Frobenius norm.
Hence, one can use a traditional plug-in estimator for~$ \eta_i(\mathbf D_{1:n})$ by replacing in~\eqref{eta_i}~$(i)$ the envelope function~$\mathbf p$ by~$\widehat {\mathbf p}$ from~\eqref{env-hat},~$(ii)$ the pairwise distances by their estimates~$\left(\widehat G_{i,j} \right)_{1\leq i,j \leq n}$ and~$(iii)$ the latitude function~$f_{\mathcal L}$ by a non-parametric kernel density estimator built from the latent distances between consecutive nodes~$\left(\langle X_i, X_{i+1}\rangle \right)_{i\in[n-1]}$ estimated by~$\left(\widehat G_{i,i+1}\right)_{i \in [n-1]}$.

\bigskip
Through the example of MRGG, one can easily grasp the interest of growth model for random graphs with a geometric structure. Modeling the time evolution of networks, one can hope to solve tasks such as link prediction or collaborative filtering. An interesting research direction would be to extend the previous work to an anisotropic Markov kernel.

\section{Connections with community based models}
\label{sec:community-RGG}

We have already described open problems and interesting directions to pursue regarding the questions tackled in the Sections~\ref{sec:detecting-geometry}, \ref{sec:non-parametric} and \ref{sec:MRGG}. In this last section, we want to look at RGGs from a different lens by highlighting a recently born line of research that investigates the connections between RGGs and community based models. Without aiming at presenting in a comprehensive manner the literature on this question, we rather focus on a few recent works that could inspire the reader to contribute in this emerging field.
\medskip

A plenty number of random graph models have been so far studied. However real world problems never match a particular model and most of the time present several internal structures. To take into account this complexity, a growing number of works have been trying to take the best of several known random graph models. \cite{Papadopoulos12} introduced a growth model where new connections with the upcoming node are drawn taking into account both popularity and similarity of vertices. The motivation is to find a balance between two trends for new connections in social networks namely {\it homophily} and {\it popularity}. 
 One can also mention \cite{Jordan15} who consider a growth model that interpolates between pure preferential attachment (essentially the well-known Barabasi–Albert model) and a purely geometric model (the online nearest-neighbour graph). As pointed out by \cite[Section II.B.3.a]{Barthelemy11}, {\it " it  is  clear  that  community  detection  in  spatial networks is a very interesting problem which might receive a specific answer."} 

\subsection{Extension of RGGs to take into account community structure}

\cite{galhotra17} proposed a new random graph model that incorporates community membership in standard RGGs. More precisely, they introduce the Geometric Block Model which is defined as follows. Consider~$V=V_1  \sqcup V_2  \sqcup \dots \sqcup V_k$ a partition of~$[n]$ in~$k$ clusters,~$\left(X_u\right)_{u \in [n]}$ independent and identical random vectors uniformly distributed on~$\mathds S^{d-1}$ and let~$\left(r_{i,j}\right)_{1\leq i , j \leq k} \in [0,2]^{k \times k}$. The Geometric Block Model is a random graph with vertices~$V$ and an edge exists between~$v\in V_i$ and~$u\in V_j$ if and only if~$\|X_u-X_v\|\leq r_{i,j}$. Focusing on the case where~$r_{i,i}= r_s, \forall i~$ and~$r_{i,j}=r_d, \; \forall i\neq j$, the authors want to recover the partition~$V$ observing only the adjacency matrix of the graph. They proved that in the relatively sparse regime (i.e. when~$r_s,r_d = \Omega_n\left(\frac{\log n }{n}\right)$), a simple motif-counting algorithm allows to detect communities in the Geometric Block Model and is near-optimal. The proposed greedy algorithm affects two nodes to the same community if the number of their common neighbours lies in a prescribed range whose bounds depend on~$r_s$ and~$r_d$ that are assumed to be known. The method is proved to recover the correct partition of the nodes with probability tending to~$1$ as~$n$ goes to~$+\infty.$

\pagebreak[3]

In \cite{sankararaman17}, the previous work is extended by considering arbitrary connection function. The paper sheds light on interesting differences between the standard SBMs and community models that incorporates some geometric structure. We start by presenting their model before highlighting some interesting results. Their model is the Planted Partition Random Connection Model (PPCM) that relies on a Poisson Point Process on~$\mathds R^d$ with intensity~$\lambda>0$~$\phi :=\{X_1, X_2, \dots \}$ where it is assumed that the enumeration of the points~$X_i$ is such that for all~$i,j \in \mathds N$,~$i>j \implies \|X_i\|_{\infty}\geq \|X_j\|_{\infty}$. Each atom~$i \in \mathds N$ is marked with a random variable~$Z_i \in \{-1,+1\}$.~$\overline{\phi}$ is the marked Poisson Point Process. The  sequence~$\{Z_i\}_{i\in \mathds N}$ is  i.i.d. with each element being uniformly distributed in~$\{-1,+1\}$. The interpretation of this marked point process is that for any node~$i\in \mathds N$, its location label is~$X_i$ and its community label is~$Z_i$. Considering two connection functions~$f_{in},\; f_{out}: \mathds R_+ \to [0,1]$, they first construct an infinite graph~$G$ with vertex set~$\mathds N$ and place an edge between any two nodes~$i,j \in \mathds N$ with probability~$f_{in}(\|X_i-X_j\|) \mathds 1_{Z_i=Z_j}+f_{out}(\|X_i-X_j\|)\mathds 1_{Z_i\neq Z_j}$. The graph~$G_n$ is then the induced subgraph  of~$G$ consisting of the nodes 1 through~$N_n$ where~$N_n:= \sup \left\{ i\geq 0 : X_i \in B_n:=[-\frac{n^{1/d}}{2},\frac{n^{1/d}}{2}]^d \right\}$. \\
Considering that the graph is observed and that the connections functions~$f_{in},f_{out}$ and the location labels~$(X_i)_i$ are known, the authors investigate conditions on the parameters of their model allowing to extract information on the community structure from the observed data.

\paragraph{Weak recovery}

Weak Recovery is said to be solvable if for every $n\in \mathds N \backslash \{0\}$, there exists some algorithm that - based on the observed data $G_n$ and $\phi$ - provides a sequence of $\{-1,+1\}$ valued random variables $\{\tau^{(n)}_i\}_{i=1}^{N_n}$ such that there exists a constant $\gamma >0$ such that the {\it overlap} between $\{\tau^{(n)}_i\}_{i=1}^{N_n}$ and $\{Z_i\}_{i=1}^{N_n}$ is asymptotically almost surely larger than $\gamma$, namely \[\lim_{n\to \infty}  \mathds P\left(\frac{\sum_{i=1}^{N_n} \tau^{(n)}_i Z_i}{N_n}\geq \gamma\right) = 1.\]
The authors identify regimes where weak recovery can be solved or not. We summarize their results with Proposition~\ref{prop:sbm-geo-weak}.
\begin{proposition} \label{prop:sbm-geo-weak} \cite[Proposition 1 - Corollary 2 - Theorem 2]{sankararaman17}

\noindent For every~$f_{in}(\cdot), f_{out}(\cdot)$ such that~$\{r \in \mathds  R^+ \; : \; f_{in}(r) \neq f_{out}(r)\}$ has positive Lebesgue
measure and any~$d\geq 2$, there exists a~$\lambda_c \in (0, \infty)$ such that
\begin{itemize}
\item for any~$\lambda<\lambda_c$,  weak recovery is not solvable.
\item for any~$\lambda > \lambda_c$, there exists an algorithm (which could possibly take exponential time) to solve
weak recovery.
\end{itemize}
Moreover, there exists~$\tilde{\lambda}_c < \infty$ (possibly larger than~$\lambda_c$) depending on~$f_{in}(\cdot), f_{out}(\cdot)$ and~$d$, such that for
all~$\lambda >\tilde{\lambda}_c~$, weak recovery is solvable in polynomial time.
\end{proposition}
The intrinsic nature of the problem of weak recovery is completely different in the PPCM model compared to the standard sparse SBM. Sparse SBMs are known to be locally tree-like with very few short cycles. Efficient algorithms that solve weak recovery in the sparse SBM (such as message passing algorithm, convex relaxation or spectral methods) deeply rely on the local tree-like structure. On the contrary, PPCMs are locally dense even if their are globally sparse. This is due to the presence of a lot of short loops (such as triangles). As a consequence, the standard tools used for SBMs are not relevant to solve weak recovery in PPCMs. Nevertheless, the local density allows to design a polynomial time algorithm that solves weak recovery for~$\lambda > \tilde{\lambda}_c$ (see Proposition~\ref{prop:sbm-geo-weak}) by simply considering the neighbours of each node. Proposition~\ref{prop:sbm-geo-weak} lets open the question of the existence of a gap between information versus computation thresholds. Namely, is it always possible to solve weak recovery in polynomial time when~$\lambda > \lambda_c$? In the sparse and symmetric SBM, it is known that there is no information-computation gap for~$k=2$ communities, while for~$k\geq 4$ a non-polynomial algorithm is known to cross the Kesten-Stigum threshold which was conjectured by \cite{decelle11} to be the threshold at which weak recovery can be solved efficiently.

\paragraph{Distinguishability}
The distinguishability problem asks how well one can solve a hypothesis testing problem that consists in finding if a given graph has been sampled from the PPCM model or from the null, which is given by a plain random connection
model with connection function~$(f_{in}(\cdot)+f_{out}(\cdot))/2$ without communities but having the same average degree and distribution for spatial locations. \cite{sankararaman17} prove that for every~$\lambda > 0$,~$d \in \mathds N$ and connection functions~$f_{in}(\cdot)$ and~$f_{out}(\cdot)$ satisfying~$1 \geq 
f_{in}(r) \geq  f_{out}(r) \geq  0$ for all~$r \geq 0$, and~$\{r \geq  0 : f_{in}(r) \neq  f_{out}(r)\}$ having positive Lebesgue measure, the probability distribution of the null and the alternative of the hypothesis test are mutually singular. As a consequence, there exists some regimes (such as~$\lambda < \lambda_c$ and~$d \geq 2$) where we can be very sure by observing the data that a partition exists, but cannot identify it better than at random. In these cases, it is out of reach to bring together the small partitions of nodes in different regions of the space into one coherent. Such behaviour does not exist in the sparse SBM with two communities as proved by \cite{mossel14} and was conjectured to hold also for~$k\geq 3$ communities in \cite{decelle11}.

\subsection{Robustness of spectral methods for community detection with geometric perturbations}
\label{subsec:robustness}

In another line of work, \cite{peche20} are studying robustness of spectral methods for community detection when connections between nodes are perturbed by some latent random geometric graph. They identify specific regimes in which spectral methods are still efficient to solve community detection problems despite geometric perturbations and we give an overview of their work in what follows. Let us consider some fixed parameter~$\kappa \in [0,1]$ that drives the balance between strength of the community signal and the noise coming from the geometric perturbations. For sake of simplicity, they consider a model with two communities where each vertex~$i$ in the network is characterized by some vector~$X_i \in \mathds R^2$ with distribution~$\mathcal N(0,I_2)$. They consider~$p_1,p_2 \in (0,1)$ that may depend on the number of nodes~$n$ with~$p_1>p_2$ and~$ \sup_n p_1/p_2 <\infty$. Assuming for technical reason~$\kappa+ \max\{p_1,p_2\}\leq 1$, the probability of connection between~$i$ and~$j$ is \[\mathds P\left\{i\sim j\;|\; X_i,X_j\right\}=\kappa \exp\left(-\gamma\|X_i-X_j\|^2\right)+\left\{
    \begin{array}{ll}
        p_1 & \mbox{if } i \text{ and } j \text{ belong to the same community}  \\
        p_2 & \mbox{otherwise.}
    \end{array}
\right.,\]
where the inverse width~$\gamma>0$ may depend on~$n$. We denote by~$\sigma\in \{\pm 1/\sqrt n\}^n$ the normalized community vector illustrating to which community each vertex belong ($\sigma_i=-1/\sqrt n$ if~$i$ belongs to the first community and~$\sigma_i=1/\sqrt n$ otherwise). The matrix of probabilities of this model is given by~$Q:= P_0+P_1$ where 
\[P_0 := \begin{bmatrix} p_1 J & p_2 J \\ p_2 J & p_1 J  \end{bmatrix}\quad \text{and} \quad P_1 :=\kappa P = \kappa \left( (1-\delta_{i,j}) e^{- \gamma \|X_i-X_j\|^2} \right)_{1\leq i,j\leq n}.\] The adjacency matrix~$A$ of the graph can thus be written as~$A=P_0+P_1+A_c$ where~$A_c$ is, conditionnally on the~$X_i$’s, a random matrix with independent Bernoulli entries which are centered. Given the graph-adjacency matrix $A$, the objective is to output a normalized vector $x\in \{\pm 1/ \sqrt n \}^n$ such that, for some $\epsilon>0$,
\begin{itemize}
\item Exact recovery: with probability tending to $1$, $|\sigma^{\top}x|= 1,$ \item Weak recovery (also called detection): with probability tending to 1, $|\sigma^{\top}x|>\epsilon.$
\end{itemize}
Let us highlight that contrary to the previous section, the latent variables~$(X_i)_i$ are
not observed. When~$\kappa=0$, we recover the standard SBM:~$Q=P_0$ has two non zero eigenvalues which are~$\lambda_1=n(p_1+p_2)/2$ with associated normalized eigenvector~$v_1=\frac{1}{\sqrt n}(1,1,\dots, 1)^{\top}$ and~$\lambda_2=n(p_1-p_2)/2$ associated to~$v_2=\sigma=\frac{1}{\sqrt n}(1,\dots,1,-1,\dots, -1)^{\top}$.  Spectral methods can thus be used to recover communities by computing the second eigenvector of the adjacency matrix~$A$. To prove that spectral methods still work in the presence of geometric perturbations, one needs to identify regimes in which the eigenvalues of~$A$ are well separated and the second eigenvector is approximately~$v_2.$ 

In the regime where~$\gamma \gg n/\log n$, the spectral radius~$\rho(P_1)$ of~$P_1$ vanishes and we asymptotically recover a standard SBM. Hence, they focus on the following regime
\begin{equation}
\gamma \underset{n \to \infty}{\to} \infty \quad \text{and} \quad \frac{1}{\gamma} \frac{n}{\ln n }\underset{n \to \infty}{\to} \infty. \label{A1} \tag{$A_1$}
 \end{equation}
Under Assumption~\eqref{A1}, \cite[Proposition 2]{peche20} states that with probability tending to one,~$\rho(P_1)$ is of order~$\frac{\kappa n}{2\gamma}.$ Using \cite[Theorem 2.7]{benaych20} to get an asymptotic upper-bound on the spectral radius of~$A_c$, basic perturbation arguments would prove that standard techniques for community detection work in the regime where 
\[\frac{\kappa n}{2\gamma}\ll \sqrt{\frac{n(p_1+p_2)}{2}} = \sqrt{\lambda_1}.\]
Indeed, it is now well-known that weak recovery in the SBM can be solved efficiently as soon as~$\lambda_2 >  \sqrt{\lambda_1}$ (for example using the power iteration algorithm on the non-backtracking matrix from \cite{bordenave2015}).
Hence, the regime of interest correspond to the case where 
\begin{equation}\exists c,C>0 \quad \text{s.t.} \quad \lambda_2^{-1} \frac{\kappa n}{2\gamma} \in [c,C], \quad \frac{\lambda_2}{\lambda_1} \in [c,C] \quad \text{and} \quad \lambda_2 \gg \sqrt{\lambda_1},     \label{A2} \tag{$A_2$} \end{equation}
which corresponds to the case where the noise induced by the latent random graph is of the same order of magnitude as the signal.  Under~\eqref{A2}, the problem of weak recovery can be tackled using spectral methods on the matrix $S=P_0+P_1$: the goal is to reconstruct communities based on the second eigenvector of~$S$. To prove that these methods work, the authors first find conditions ensuring that two eigenvalues of~$S$ exit the support of the spectrum of~$P_1$. Then, they provide an asymptotic lower bound for the level of correlation between~$v_2=\sigma$ and the second eigenvector~$w_2$ of~$S$, which leads to Theorem~\ref{mainthm-peche20}.
\begin{theorem} \cite[Theorem 10]{peche20} \label{mainthm-peche20}

\noindent Suppose that Assumptions~\eqref{A1} and~\eqref{A2} hold and that~$\lambda_1> \lambda_2+ 2 \frac{\kappa}{2\gamma}$. Then the correlation~$|w^{\top}_2v_2|$ is  uniformly  bounded  away  from ~$0$. Moreover, denoting~$\mu_1$ the largest eigenvalue of~$P_1$, if the ratio~$\lambda_2/\mu_1$ goes to infinity then~$|w^{\top}_2v_2|$ tends to~$1$, which gives weak (and even exact at the limit) recovery.
\end{theorem}

\subsection{Recovering latent positions}
\label{subsec:latent-pos}

From another viewpoint, one can think RGGs as an extension of stochastic block models where the discrete community structure is replaced by an underlying geometry. With this mindset, it is natural to directly transport concepts and questions from clustered random graphs to RGGs. For instance, the task consisting in estimating the communities in SBMs may correspond to the estimation of latent point neighborhoods in RGGs. More precisely, community detection can be understood in RGGs as the problem of recovering the geometric representation of the nodes (e.g. through the Gram matrix of the latent positions). This question has been tackled by \cite{eldan2020community} and \cite{araya20}. Both works consider random graphs sampled from the TIRGG model on the Euclidean sphere $\mathds S^{d-1}$ with some envelope function $\mathbf p$ (see Definition~\ref{def:tirgg}), leading to a graphon model similar to the one presented in Section~\ref{sec:graphon}. While the result from \cite{araya20} holds in the dense and relatively sparse regimes, the one from \cite{eldan2020community} covers the sparse case. Thanks to harmonic properties of $\mathds S^{d-1}$, the graphon eigenspace composed only with linear eigenfunctions (harmonic polynomials of degree one) directly relates to the pairwise distances of the latent positions. This allows \cite{eldan2020community} and \cite{araya20} to provide a consistent estimate of the Gram matrix of the latent positions in Frobenius norm using a spectral method. Their results hold under the following two key assumptions. 
\begin{enumerate}
\item {\it An eigenvalue gap condition.} They assume that the $d$ eigenvalues of the integral operator~$\mathds T_W$ - associated with the graphon $W:= \mathbf p(\langle \cdot,\cdot \rangle)$ (see \eqref{eq:inte-ope}) - corresponding to the Spherical Harmonics of degree one is well-separated from the rest of the spectrum.
\item {\it A regularity condition.} They assume that the envelope function $\mathbf p$ belongs to some Weighted Sobolev space, meaning that the sequence of eigenvalues of $\mathds T_W$ goes to zero fast enough.
\end{enumerate}
In addition to similar assumptions, \cite{eldan2020community} and~\cite{araya20} share the same proof structure. First they need to recover the $d$ eigenvectors from the adjacency matrix corresponding to the space of spherical Harmonics of degree one. Then the Davis-Kahan Theorem is used to prove that the estimate of the Gram matrix based on the previously selected eigenvectors is consistent in Frobenius norm. To do so, they require a concentration result ensuring that the adjacency matrix $A$ (or some proxy of it) converges in operator norm towards the matrix of probabilities~$\Theta$ with entries~$\Theta_{i,j}=\mathbf p\left(\langle X_i,X_j\rangle\right)$ for~$1\leq i \neq j\leq n$ and zero diagonal entries. \cite{araya20} relies on~\cite[Corollary 3.12]{bandeira16}, already discussed in~\eqref{bandeira:iid-RGG}, that provides the convergence $\|A-\Theta\|\to 0$ as $n\to \infty$ in the dense and relatively sparse regimes. In the sparse regime, such concentration no longer holds. Indeed, in that case, degrees of some vertices are much higher than the expected degree, say $\mathrm{deg}$. As a consequence, some rows of the adjacency matrix $A$ have Euclidean norms much larger than $\sqrt{\mathrm{deg}}$, which implies that for $n$ large enough, it holds with high probability $\|A - \Theta\|\gg \sqrt{\mathrm{deg}}$. To cope with this issue,~\cite{eldan2020community} do not work directly on the adjacency matrix but rather on a slightly amended version of it - say $A'$ - where one reduces the weights of the edges incident
to high degree vertices. In that way, all
degrees of the new (weighted) network become bounded, and~\cite[Theorem 5.1]{le2018} ensures that $A'$ converges to $\Theta$ in spectral norm as $n$ goes to $+\infty.$ Hence in the sparse regime the adjacency matrix converges towards its expectation {\it after regularization}. The proof of this random matrix theory tool is based on a famous result in functional analysis known as the Grothendieck-Pietsch factorization. \\
Let us finally mention that this change of behaviour of the extreme eigenvalues
of the adjacency matrix according to the maximal mean degree has been studied in details for inhomogeneous Erdös-Rényi graphs in~\cite{benaych20} and~\cite{benaych2019}.

\subsection{Some perspectives}

The paper \cite{sankararaman17} makes the strong assumption that the locations labels~$(X_i)_{i\geq 1}$ are known. Hence it should be considered as an initial work calling for future theoretical and practical investigations. Keeping the same model, it would be of great interest to design algorithms able to deal with unobserved latent variables to allow real-data applications. A first step in this direction was made by \cite{Avrachenkov_2021} where the authors propose a spectral method to recover hidden clusters in the Soft Geometric Block
Model where latent positions are not observed. On the theoretical side, \cite{sankararaman17} describe at the end of their paper several open problems. Their suggestions for future works include~$i)$ the extension of their work to a larger number of communities,~$ii)$ the estimation from the data of the parameters of their model (namely~$f_{in}$ and~$f_{out}$ that they assumed to be known), and~$iii)$ the existence of a possible gap between information versus computation thresholds, namely, they wonder if there is a regime where community detection is solvable, but without any polynomial (in~$n$) time and space algorithms. \\
Another possible research direction is the extension of the work from Section~\ref{subsec:robustness} to study the same kind of robustness results for more than 2 communities and especially in the sparse regime where $\frac{1}{\gamma} \sim p_i \sim \frac{1}{n}$. As highlighted by~\cite{peche20}, the sparse case may bring additional difficulties since {\it " standard spectral techniques in this regime involve the non-backtracking matrix (see~\cite{bordenave2015}), and its concentration properties are quite challenging to establish."} Regarding Section~\ref{subsec:latent-pos}, for some applications it may be interesting to go beyond the recovery of the pairwise distances by embedding the graph in the latent space while preserving the Gram structure. Such question has been tackled for example by \cite{perry20} but only for the Euclidean sphere in small dimensions.

\newpage

\newpage

\paragraph{Acknowledgements}

The authors are in debt to Tselil Schramm who gave a great talk at the S.S.Wilks Memorial Seminar in Statistics (at Princeton University) providing insightful comments on the problem of geometry detection or more specifically on her paper \cite{Liu2021TestingTF}.

\appendix
\section{Outline of the proofs of Theorems~\ref{thm:polylog} and~\ref{thm:polylog2}}
\label{apdx}

The proofs of Theorems~\ref{thm:polylog} and~\ref{thm:polylog2} (cf. Section~\ref{sec:resolution-geo}) are quite complex and giving their formal descriptions would require heavy technical considerations. In the following, we provide an overview of the proofs highlighting the nice mathematical tools used by \cite{Liu2021TestingTF} and their innovative combination while putting under the rug some technical aspects.  
\begin{enumerate}
\item[Step 1.] Relate the TV distance of the whole graphs to single vertex neighbourhood.
\begin{align}
2 \mathrm{TV}(G(n,p,d),G(n,p))^2 &\leq \mathrm{KL}(G(n,p,d)||G(n,p))\qquad \text{from Pinsker's inequality}\notag\\
&\leq n\times  \mathds E_{G_{n-1} \sim G(n-1,p,d)}\left[ \mathrm{KL}\left( \nu_n(\cdot|G_{n-1}) ,\mathrm{Bern}(p)^{\otimes(n-1)} \right) \right]\qquad \text{from Lemma~\ref{brennan20-lemmaKL}}\notag\\
&= \mathds E_{G_{n-1} \sim G(n-1,p,d)}\mathds E_{S \sim \nu_n(\cdot|G_{n-1})} \log \left(\frac{\nu_n(S|G_{n-1})}{p^{|S|}(1-p)^{n-1-|S|}}\right),\label{eq:tv2proba}
\end{align}
where $\nu_{n}(\cdot|G_{n-1})$ denotes the distribution of the neighbourhood of vertex $n$ when the graph is sampled from $G(n,p,d)$ conditional on the knowledge of the connections between pairs of nodes in $[n-1]$ given by $G_{n-1}$. Hence, the main difference with \cite{brennan20} is that the tensorization argument from Lemma~\ref{brennan20-lemmaKL} is used node-wise (and not edge-wise). We are reduced to understand how a vertex incorporates a given graph of size $n-1$ sampled from the distribution $G(n-1,p,d)$. At a high level, the authors show that if one can prove that for some $\epsilon>0$, with high probability over $G_{n-1} \sim G(n-1,p,d)$, it holds 
\begin{equation}\label{eq:neighb-probas}\forall S \subseteq [n-1],\quad \nu_n(S|G_{n-1})=\mathds P_{G \sim G(n,p,d)} ( N_G(n) = S \, |\, G_{n-1})=(1\pm \epsilon) p^{|S|} (1-p)^{n-1-|S|},\end{equation}
where $N_G(n)$ denotes the set of nodes connected to node $n$ in the graph $G$, then
\begin{equation}\label{eq:liu-proof1}\mathrm{TV}(G(n,p,d),G(n,p)) =o_n(n \epsilon^2).\end{equation}
\item[Step 2.] Geometric interpretation of neighbourhood probabilities from Eq.\eqref{eq:neighb-probas}.\\
For $G\sim G(n, p,d)$, if vertex $i$ is associated to a (random) vector $X_i$, and $(i, j)$ is an edge, we consequently know that $\langle X_i
,X_j\rangle \geq t_{p,d}$. On the sphere $\mathds S^{d-1}$, the locus of points where $X_j$ can be, conditioned
on $(i, j)$ being an edge, is a sphere cap centered at $X_i$ with a $p$ fraction of the sphere’s surface area, which we denote by $\mathrm{cap}(X_i)$. Similarly, if we know that $i$ and $j$ are not adjacent, the locus of points where $X_j$ can fall is the complement of a sphere cap with measure $1-p$ namely $\overline{\mathrm{cap}(X_i)}$, which we call an “anti-cap”.
Let us denote $\sigma$ is the normalized Lebesgue measure on $\mathds S^{d-1}$ so that $\sigma(\mathds S^{d-1})=1$. Equipped with this geometric picture, we can view the probability that vertex $n$’s neighborhood is exactly equal to $S \subseteq [n -1]$ as $\sigma(L_S)$, where $L_S \subseteq \mathds  S^{d-1}$ is a random set defined by \[L_S := \big(\bigcap_{i\in S} \mathrm{cap}(X_i)\big) \cap \big( \bigcap_{j\notin S} \overline{\mathrm{cap}(X_j)}\big).\]
To show that the TV distance between $G(n, p, d)$ and $G(n, p)$ is small, we need to prove that $\sigma(L_S)$ concentrates around $p^{|S|}(1-p)^{n-1-|S|}$ as suggested by Eqs.\eqref{eq:neighb-probas} and \eqref{eq:liu-proof1}.
\item[Step 3.] Concentration of measure of intersections of sets in $\mathds S^{d-1}$ with random spherical caps. \\ An essential contribution of \cite{Liu2021TestingTF} is a novel concentration inequality for the area of the intersection of a random spherical cap with any subset $L\subseteq \mathds S^{d-1}$.
\begin{Lemma} \citep[see][Corollary 4.10]{Liu2021TestingTF} \label{lemma:cap-cap} {\bf Set-cap intersection concentration Lemma.}\\ 
Suppose $L\subseteq \mathds S^{d-1}$ and let us denote by $\sigma$ the uniform probability measure on $\mathds S^{d-1}$. Then with high probability over $z\sim \sigma$ it holds
\[ \big|\frac{\sigma(L\cap \mathrm{cap}(z))}{p\sigma(L)}-1\big| = \mathcal O_n\big(\delta_n(L)\big) \quad \text{and}\quad  \big|\frac{\sigma(L\cap \overline{\mathrm{cap}(z)})}{(1-p)\sigma(L)}-1\big| = \mathcal O_n\big( \frac{p}{1-p}\delta_n(L)\big),\]
where $\delta_n(L)= \sqrt{\frac{\log \frac1p + \log \frac{1}{\sigma(L)} }{\sqrt d}}\mathrm{polylog}(n)$.
\end{Lemma}
\begin{proof}[Sketch of proof of Lemma~\ref{lemma:cap-cap}.]
We give an overview of the proof of Lemma~\ref{lemma:cap-cap}, highlighting its interesting connection with optimal transport. Let us consider some probability distribution $\nu$ on $\mathds S^{d-1}$. Let us denote $\mathcal D$ the optimal coupling between the measures $\nu$ and $\sigma$, i.e. $\mathcal D$ is a probability measure on $\mathds S^{d-1}\times \mathds S^{d-1}$ with marginals $\nu$ and $\sigma$ such that 
\[W_2(\nu,\sigma)^2= \int \|x-y\|_2^2 d\mathcal D(x,y),\]
where $W_2(\nu,\sigma)$ is the Wasserstein 2-distance between the measures $\sigma$ and $\nu$. Then for any $z\in \mathds S^{d-1}$ it holds
\begin{align}
\mathds P_{x \sim \nu}(\langle z,x\rangle >t_{p,d})=&\mathds P_{(x,y) \sim \mathcal D}(\langle z,y\rangle >t_{p,d}-\langle z,x-y\rangle)\notag\\
\leq & \mathds P_{y \sim \sigma}(\langle z,y\rangle >t_{p,d}-u(p,d))+\mathds P_{(x,y)\sim \mathcal D}(|\langle z,x-y\rangle |>u(p,d)),\label{eq:OT-lemma}
\end{align}
for some well chosen threshold $u(p,d)$ depending on $p$ and $d$. The first term in the right hand side of Eq.\eqref{eq:OT-lemma} can be proven to concentrate around $p$ with high probability over $z\sim \sigma$ with standard arguments. The second term in Eq.\eqref{eq:OT-lemma} quantifies how often a randomly chosen transport vector $x-y$ with $(x,y)\sim \mathcal D$ has a large projection in the direction $z$. One can prove that the optimal transport map $\mathcal D$ between $x \sim \nu$ and $y \sim \sigma$ has bounded length with high probability, and then translate this into a tail bound for the inner product
$\langle z, x -y\rangle$ for a random vector $z\sim \sigma$. As a consequence, one can bound with high probability over $z \sim \sigma$ the fluctuations of $\left|\mathds P_{x \sim \nu}(\langle z,x\rangle >t_{p,d})-p\right|$ which gives Lemma~\ref{lemma:cap-cap} if we take for $\nu$ the uniform measure on the set $L\subseteq \mathds S^{d-1}.$
\end{proof}
Applying Lemma~\ref{lemma:cap-cap} inductively and using a martingale argument, the authors prove that intersecting
$j$ random caps and $(k - j)$ random anticaps, we get a multiplicative fluctuation for $\sigma(L_S)$ around $p^{|S|}(1-p)^{n-1-|S|}$ that is of the order of $(1\pm \sqrt j \delta+ \sqrt{k-j}\frac{p}{1-p}\delta)$. Going back to Eq.\eqref{eq:liu-proof1}, this approach is sufficient to prove that \[\mathrm{TV}(G(n,p,d),G(n,p))= o_n\big( \frac{n^3p^2}{d} \big),\]
leading to the first statement of Theorem~\ref{thm:polylog2}.
\item[Step 4.] The sparse case and the use of the cavity method.\\
To get down to a polylogarithmic threshold in the sparse regime, the authors change of paradigm. Previously, they were bounding the quantity 
\begin{equation}\label{eq:neighboorhoud}\mathds P_{G \sim G(n,p,d)} ( N_G(n) = S \, |\, G_{n-1}) = \mathds E_{X_1,\dots,X_{n-1} \, |\, G_{n-1}} \mathds E_{X_n \sim \sigma}\big[ \mathds 1_{ N_G(n) = S} \big]=\mathds E_{X_1,\dots,X_{n-1} \, |\, G_{n-1}} \big[\sigma (L_S) \big],\end{equation}
by fixing a specific realization of latent positions $X_1, \dots X_{n-1}$ and then analyzing the probability that
the node $n$ connects to some $S \subseteq [n - 1]$. The probability that vertex $n$ is adjacent to all vertices in $S \subseteq [n-1]$ is exactly equal to the measure of the set-caps intersection, which appears to be tight. At a high level, this is a {\it "worst case approach"} to upper bound Eq.\eqref{eq:neighboorhoud} in the sense that the bound obtained from this analysis may be due to an unlikely latent configuration conditioned on $X_1, \dots, X_{n-1}$ producing $G_{n-1}$. To obtain a polylogarithmic threshold in the sparse case, one needs to analyze the concentration of $\sigma(L_S)$ on average over vector embeddings of $G_{n-1}$. 
To do so, the authors rely on the so-called {\it cavity method} borrowed from the field of statistical physics. The cavity method allows to understand the distribution of $(X_i)_{i\in S}$ conditional on forming $G_{n-1}$ for any $S\subseteq [n-1]$ with size of the order $pn=\Theta(1)$. We provide further details on this approach in the following.
\paragraph{A simplification using tight concentration for intersections involving anti-caps.}\cite{Liu2021TestingTF} first prove that due to tight concentration for the measure of the intersection of random anticaps with sets of lower bounded measure, one can get high-probability estimates for $\nu_n(S\,|\,G_{n-1})$ by studying the probability that $S \subseteq N_G(n)$, namely
\begin{equation}\label{eq:subsetneighbourhood}\mathds P (S\subseteq N_G(n) \, | \, G_{n-1}) = \mathds P\big( \forall i \in S, \; \langle X_i,X_n\rangle \geq t_{p,d}\, | \, G_{n-1}\big) = \underset{\substack{\qquad  X_n \sim \sigma \\ (X_i)_{i\in [n-1]}} \sim \sigma^{G_{n-1}}}{\mathds E} \prod_{i \in S} \mathds 1_{\langle X_i,X_n \rangle \geq t_{p,d}},\end{equation}
where $\sigma^{G_{n-1}}:=\left[\sigma^{\otimes (n-1)} \,|\, G_{n-1}\right]$. If $(X_i)_{i\in S}$ in Eq.\eqref{eq:subsetneighbourhood} was a collection of \underline{independent} random vectors distributed \underline{uniformly} on the sphere then Eq.\eqref{eq:subsetneighbourhood} would be exactly equal to $p^{|S|}$. In the following, we explain how the authors prove that both
of these properties are approximately true.
\paragraph{The cavity-method.} To bound the fluctuation of Eq.\eqref{eq:subsetneighbourhood} around $p^{|S|}$, \cite{Liu2021TestingTF} use the cavity-method. Let us consider $S\subseteq [n-1]$, $G_{n-1}$ sampled from $G(n-1,p,d)$ and its corresponding latent vectors. Let us denote by $\mathcal B_{G_{n-1}}(i,\ell)$ the ball of radius-$\ell$ around a vertex $i\in[n-1]$ in the graph $G_{n-1}$. Fixing all vectors except those in $K:= \bigcup_{i \in S} \mathcal B_{G_{n-1}}(i,\ell-1)$, the cavity method aims at computing the joint distribution of $(X_i)_{i\in S}$ conditional to $(X_i)_{i\notin K}$ and $G_{n-1}$. Informally speaking, we "carve out" a cavity of depth $\ell$ around each vertex $i \in S$ and we fix all latent vectors outside of these cavities as presented with Figure~\ref{fig:cavity-method}. 
\begin{figure}[!ht]
\centering
\includegraphics[scale=0.4]{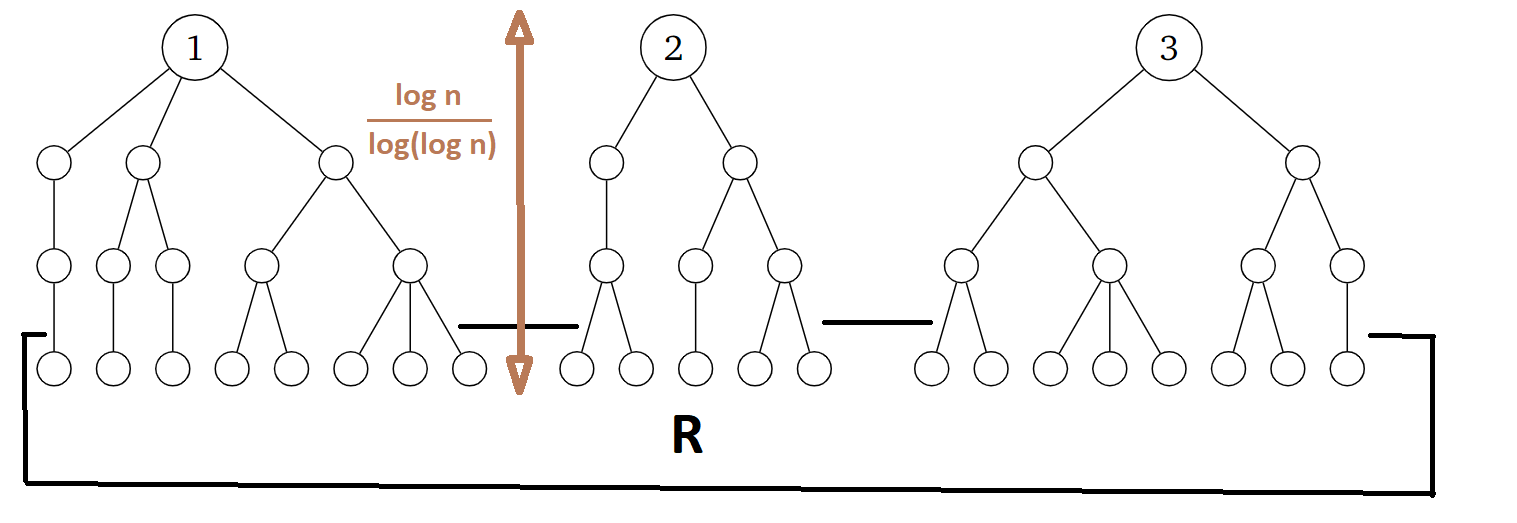}
\caption{Illustration of the cavity method to bound the fluctuation of Eq.\eqref{eq:subsetneighbourhood} around $p^{|S|}$ {\it i.e.,} to bound the deviation of the random variable $\sigma(L_S)$ conditioned on $X_1,\dots, X_{n-1}$ producing $G_{n-1}$. With high probability, the neighbourhood until depth $\ell = \frac{\log n}{\log \log n}$ of vertices in $S$ are disjoint trees. We fix the latent representation of vertices in the set $R:=[n-1] \backslash K$. Using the Belief Propagation algorithm, one can compute the distribution of $(X_i)_{i\in S} \;|\; (X_j)_{j \in R}$ where the latent positions $(X_j)_{j \in [n-1]}$ are sampled according to $\sigma^{G_{n-1}}$. This allows to bound the fluctuation of Eq.\eqref{eq:subsetneighbourhood} around $p^{|S|}$.}
\label{fig:cavity-method}
\end{figure}
The choice of the depth $\ell$ results from the following tradeoff: 
\begin{itemize}
\item We want to choose the depth $\ell$ small enough so that the balls $\mathcal B_{G_{n-1}}(i,\ell)$ for $i \in S$ are all trees and are pairwise disjoint with high probability.
\item We want to choose $\ell$ as large as possible in order to get a bound on the fluctuations of Eq.\eqref{eq:subsetneighbourhood} around $p^{|S|}$ as small as possible.
\end{itemize}
To formally analyze the distribution of the unfixed vectors upon resampling them, the authors set up a constraint satisfaction problem instance over a continuous alphabet that encodes the edges of $G_{n-1}$ within the trees around $S$: each node has a vector-valued variable in $\mathds S^{d-1}$, and the constraints are
that nodes joined by an edge must have vectors with inner product at least $t_{p,d}$. The marginal of the latent vectors $X_i$ for $i\in S$ can be obtained using the Belief-Propagation algorithm. Let us recall that Belief-Propagation computes marginal distributions over labels of constraints satisfaction problems when the constraints graph is a tree. 
\paragraph{A simple analysis of Belief-Propagation.} 
To ease the reasoning, let us suppose that $1\in S$ is such that $\mathcal B_{G_{n-1}}(1,\ell-1)$ is a path. Without loss of generality, we consider that the path is given by Figure~\ref{fig:pathBP}.
\begin{figure}[!ht]
\centering
\begin{tikzpicture}
    { [start chain,
        node distance = 7mm,
        every node/.style = {inner sep = 2mm, on chain, circle,draw}, 
        every join/.style = {-},every node/.append style={join}] 
    \node (1) {1};
    \node (2) {2};
    \node (3) {3};
    \node (4) {4};
    \node (dots) {$\dots$};
    \node (e) {$\ell-1$};
    \node (ef) {$\ell$};
    \draw[-.] (9.2,-0.5) -- (9.2,1);}
    \draw node[above=of e,yshift=-0.9cm,xshift=0.4cm] (K) {$K$};
    \draw node[draw=none,right=of K,xshift=-0.3cm] {$R$};
\end{tikzpicture}
\caption{Simple analysis of the Belief Propagation algorithm when the neighbourhood of vertex $1\in S$ at depth $\ell$ is a path.}
\label{fig:pathBP}
\end{figure}
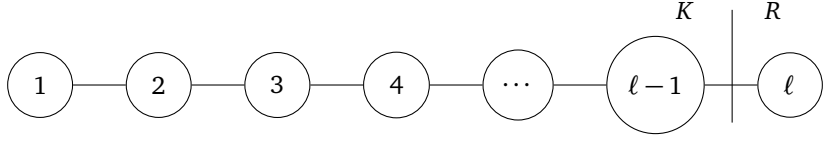
Every vector is passing to its parent along the path a convolution of its own measure (corresponding to its "message") with a cap of measure $p$. Denoting by $P$ the linear operator defined so that for any function $h : \mathds S ^{d-1}\to \mathds R$, \[Ph(x) =\frac1p \int_{\mathrm{cap}(x)} h(y) d\sigma(y),\] the authors prove that for some $a>0$, for any probability measure $\mu$ on $\mathds S^{d-1}$ with density $h$ with respect to $\sigma$, 
\begin{equation}\label{eq:contraction}\mathrm{TV}(Ph,\sigma) \leq \mathcal O_n\big(\frac{\log^a n}{\sqrt d}\big) \mathrm{TV}(\mu,\sigma),\end{equation}which is a contraction result. Since at every step of the Belief Propagation algorithm, a vertex sends to its parent the image by the operator $P$ of its own measure, we deduce from Eq.\eqref{eq:contraction} that the parent receives a measure which is getting closer to the uniform distribution by a multiplicative factor equal to $\frac1{\sqrt d}$. The proof of Eq.\eqref{eq:contraction} relies on the set-cap intersection concentration result (see Lemma~\ref{lemma:cap-cap}). To get an intuition of this connection, let us consider that $h$ is the density of the uniform probability measure $\mu$ on some set $L\subseteq \mathds S^{d-1}$, then
\[Ph(x) = \frac1p \mathds P_{Y \sim \mu}(Y \in \mathrm{cap}(x))=\frac1p \frac{\sigma(L\cap \mathrm{cap}(x))}{\sigma(L)},\]
and we can conclude using Lemma~\ref{lemma:cap-cap} that ensures that with high probability over $x \sim \sigma$, $\sigma(L\cap \mathrm{cap}(x)) = (1\pm \mathcal O_n(\frac{\log ^a n}{\sqrt d}))p\sigma(L)$. Applying Eq.\eqref{eq:contraction} $\ell = \frac{\log n}{\log \log n}$ times for $d$ being some power of $\log n$, one can show that,
\[\mathrm{TV}(P^{\ell}\mu,\sigma)=\mathcal O_n\big[\big(\frac{\log^a n}{\sqrt d}\big)^{\ell}\big] =  o_n\big(\frac{1}{\sqrt n}\big).\]
With this approach, one can prove that the distribution of $(X_i)_{i\in S}$ is approximately $\sigma^{\otimes |S|}$. This allows to bound the fluctuations of Eq.\eqref{eq:subsetneighbourhood} around $p^{|S|}$ which leads to Theorem~\ref{thm:polylog} using Eqs.\eqref{eq:neighb-probas} and \eqref{eq:liu-proof1}.\\ \\
As a concluding remark, we mention that \cite{Liu2021TestingTF} demonstrate a coupling of $G_- \sim G(n, p - o_n(p))$, $G \sim G(n, p,d)$, and $G_+ \sim G(n, p + o_n(p))$ that satisfies $G_- \subseteq G \subseteq G_+$ with high probability. This sandwich-type result holds for a proper choice of the latent dimension and allows to transfer known properties of Erdös-Renyi random graphs to RGGs in the studied regime. For example, the authors use this coupling result to upper bound the probability that the depth-$\ell$ neighborhood of some $i \in[n]$ forms a tree under $G(n, p,d)$ in the sparse regime with $d=\mathrm{polylog}(n)$.   
\end{enumerate}
\clearpage

\bibliography{sample.bib}

\end{document}